\documentclass{llncs}
\usepackage{amsmath}
\usepackage{amsfonts}
\usepackage{amssymb}
\usepackage{braket}
\usepackage{booktabs}

\usepackage{algorithm}
\usepackage{algpseudocode}
\usepackage{algorithmicx}
\usepackage{tikz}
\usepackage{pgfplots}
\usepackage{multirow}
\usepackage{enumerate}

\usepackage{xcolor}
\usepackage{xifthen}

\newcommand{\zo}{\{0, 1\}}

\newcommand{\bigO}[1]{\mathcal{O}
\ifthenelse{\isempty{#1}}{}{\! \left( #1 \right)}}
\newcommand{\bigOt}[1]{\widetilde{\mathcal{O}}
\ifthenelse{\isempty{#1}}{}{\! \left( #1 \right)}}

\newcommand{\pr}{\mathbf{Pr}} 
\newcommand{\randfrom}{\xleftarrow{~\$~}}
\newcommand{\adv}{\mathbf{Adv}} 
\newcommand{\secnot}[1]{\textsf{#1}} 
\newcommand{\alginput}{\State \textbf{input:} }
\newcommand{\algoutput}{\State \textbf{output:} }

\newcommand{\floor}[1]{\left\lfloor #1 \right\rfloor}

\usepackage{tikz}
\usetikzlibrary{shapes.geometric,shapes.misc,positioning}
\tikzstyle{block}=[draw,minimum size=2em]
\tikzstyle{l}=[minimum size=1.5em]
\tikzstyle{xor}=[draw, minimum size=0.8em,append after command={[shorten >=\pgflinewidth, shorten <=\pgflinewidth,] 
        (\tikzlastnode.north) edge (\tikzlastnode.south)
        (\tikzlastnode.east) edge (\tikzlastnode.west)
}]
\tikzstyle{sponge}=[rectangle, rounded corners=.25cm, minimum width=.5cm, minimum height=1.8cm, draw]

\usepackage[colorlinks]{hyperref}

\def\subsectionautorefname{Section}
\def\sectionautorefname{Section}

\def\subsectionautorefname{Section}

\newcommand{\AS}[1]{\textcolor{red}{[{\bf Andr\'{e} S:} #1]}}
\newcommand{\FS}[1]{\textcolor{red}{[{\bf Ferdinand S:} #1]}}

\renewcommand{\AS}[1]{}
\renewcommand{\FS}[1]{}

\newif\ifeprint
\eprinttrue

\renewcommand{\appendixname}{%
\ifeprint%
Appendix%
\else%
Supplementary Material%
\fi%
}
\newcommand\appendixref[1]{{\let\subsectionautorefname\appendixname\let\sectionautorefname\appendixname\autoref{#1}}}

\title{Beyond quadratic speedups in quantum attacks on symmetric schemes}
\ifeprint
\author{Xavier Bonnetain\inst{1,2} \and André Schrottenloher\inst{3} \and Ferdinand Sibleyras\inst{4}}
\institute{Institute for Quantum Computing, Department of Combinatorics and Optimization, University of Waterloo, Waterloo, ON, Canada \and Université de Lorraine, CNRS, Inria, Nancy, France\and Cryptology Group, CWI, Amsterdam, The Netherlands \and NTT Social Informatics Laboratories}
\else
\fi

\begin{document}

\maketitle
\renewcommand{\labelitemi}{$\bullet$}

\begin{abstract}
In this paper, we report the first quantum key-recovery attack on a symmetric block cipher design, using classical queries only, with a more than quadratic time speedup compared to the best classical attack.

We study the 2XOR-Cascade construction of Ga{\v{z}}i and Tessaro (EUROCRYPT~2012). It is a key length extension technique which provides an $n$-bit block cipher with $\frac{5n}{2}$ bits of security out of an $n$-bit block cipher with $2n$ bits of key, with a security proof in the ideal model. We show that the offline-Simon algorithm of Bonnetain et al. (ASIACRYPT~2019) can be extended to, in particular, attack this construction in quantum time $\bigOt{2^n}$, providing a 2.5 quantum speedup over the best classical attack.

Regarding post-quantum security of symmetric ciphers, it is commonly assumed that doubling the key sizes is a sufficient precaution. This is because Grover's quantum search algorithm, and its derivatives, can only reach a quadratic speedup at most. Our attack shows that the structure of some symmetric constructions can be exploited to overcome this limit. In particular, the 2XOR-Cascade cannot be used to generically strengthen block ciphers against quantum adversaries, as it would offer only the same security as the block cipher itself.
\end{abstract}



\section{Introduction}\label{section:introduction}

In~1994, Shor~\cite{DBLP:conf/focs/Shor94} designed polynomial-time quantum algorithms for factoring and computing discrete logarithms, both believed to be classically intractable. This showed that a large-scale quantum computer could break public-key cryptosystems based on these problems, such as RSA and ECC, which unluckily are the most widely used to date.

The impact of quantum computers on secret-key cryptography is, at first sight, much more limited. The lack of \emph{structure} in secret-key cryptosystems seems to defeat most exponential quantum speedups. It can be expected to do so, as it was shown in~\cite{DBLP:journals/jacm/BealsBCMW01} that relatively to an oracle, quantum speedups for worst-case algorithms can be polynomial at most, unless the oracle satisfies some additional structure. This structure, that is primordial for exponential speedups, is usually known as a ``promise''. For example, in Shor's abelian period-finding algorithm~\cite{DBLP:conf/focs/Shor94}, the promise is that the oracle is a periodic function.



Another well-known quantum algorithm, Grover's quantum search~\cite{DBLP:conf/stoc/Grover96}, can speed up an exhaustive key search by a quadratic factor. That is, an attacker equipped with a quantum computer can find the $\kappa$-bit key of a strong block cipher in about $\bigO{2^{\kappa/2}}$ operations instead of the $\bigO{2^\kappa}$ trials necessary for a classical attacker. Despite being merely polynomial, this is already an interesting advantage for this hypothetical attacker. Due to Grover's search, symmetric cryptosystems are commonly assumed to retain roughly half of their classical bits of security, and it is recommended to double their key length when aiming at post-quantum security~\cite{NAP25196}.

\paragraph{Superposition Attacks.}
In~\cite{DBLP:conf/isit/KuwakadoM10}, Kuwakado and Morii designed a polynomial-time \emph{quantum distinguisher} on the three-round Luby-Rackoff construction, although it has a classical security proof. Later on, they showed a polynomial-time key-recovery attack on the Even-Mansour construction~\cite{DBLP:conf/isita/KuwakadoM12}, a classically proven block cipher constructed from a public permutation~\cite{DBLP:journals/joc/EvenM97}.

Both of these attacks can assume ideal building blocks (random functions in the case of the Luby-Rackoff construction, a random permutation in the case of Even-Mansour), as they focus on the \emph{algebraic structure} of the construction. The target problem (distinguishing or key-recovery) is simply reduced to the problem of finding the hidden period of a periodic function, which can be solved efficiently.

However, in order to run these attacks, the quantum adversary needs to access the construction as a \emph{quantum oracle} (in \emph{superposition}). This means that the black-box must be part of a quantum circuit. When it comes to provable security in the quantum setting, this is a natural assumption, followed by most of the works in this direction (see~\cite{DBLP:conf/pqcrypto/AnandTTU16,DBLP:conf/crypto/0001Y17} for instance). However, it does not seem too hard to avoid quantum queries at the implementation level\footnote{Though it seems also impossible in some restricted cases, for example \emph{white-box encryption}. Here the adversary tries to recover the key of a block cipher whose specification is completely given to him. He can realize the quantum oracle using this specification.}.

Many other symmetric constructions have been shown to be broken under superposition queries in the past few years~\cite{DBLP:conf/crypto/KaplanLLN16,DBLP:conf/sacrypt/Bonnetain17,DBLP:conf/asiacrypt/Leander017,DBLP:conf/sacrypt/BonnetainNS19}. All these attacks have been exploiting the algebraic structure of their targets in similar ways, using different period or shift-finding algorithms.



\paragraph{Attacks based on Quantum Search.}
\emph{Quantum search}, the equivalent of classical exhaustive search, is a very versatile tool that allows to design many algorithms beyond a mere exhaustive search of the key. However, by only combining quantum search with itself, one cannot obtain a better speedup than quadratic. More precisely\footnote{For completeness, we include a short proof of this claim in \appendixref{app:ext-qsearch}.}:

\begin{quotation}
Let $\mathcal{A}$ be a quantum algorithm, with a final measurement, that is built by combining a constant number of quantum search procedures. Let $\mathcal{T}$ be its time complexity and $\mathcal{M}$ its memory complexity. Then there exists a classical randomized algorithm $\mathcal{A}'$ that returns the same results using $\mathcal{M}$ memory and time $\bigO{\mathcal{T}^2}$.
\end{quotation}

In other words, if our only quantum algorithmic tool is quantum search, then any quantum attack admits an equivalent classical attack of squared complexity, and that uses a similar memory. In particular, if the quantum procedure goes below the exhaustive key search ($\bigO{2^{\kappa/2}}$), then the corresponding classical procedure can be expected to go below the classical exhaustive key search ($\bigO{2^\kappa}$).

\paragraph{Attacks beyond Quantum Search.}
So far, when superposition queries are forbidden, all known quantum attacks on symmetric designs (e.g., key-recovery attacks on block ciphers, forgery attacks on MACs) have only been confirmed to reach time speedups less than (or equal to) quadratic: the best that quantum search, and other extended frameworks~\cite{DBLP:journals/siamcomp/MagniezNRS11}, can offer.



At ASIACRYPT~2019, Bonnetain \emph{et al.}~\cite{DBLP:conf/asiacrypt/BonnetainHNSS19} presented new attacks on the Even-Mansour and FX block ciphers that somehow went ``beyond quantum search only''. Their algorithm combines Simon's algorithm~\cite{DBLP:journals/siamcomp/Simon97a} and quantum search, inspired by an attack of Leander and May~\cite{DBLP:conf/asiacrypt/Leander017}. In some scenarios, it allows to reach a quadratic speedup \emph{and} an asymptotic memory improvement at the same time. For example, they obtained an attack on an $n$-bit Even-Mansour cipher, with $2^{n/3}$ classical queries, in quantum time $\bigOt{2^{n/3}}$, and memory $\mathsf{poly}(n)$, instead of a classical attack with time $\bigO{2^{2n/3}}$ \emph{and memory $2^{n/3}$}.


\paragraph{Contributions of this Paper.}
In this paper, we show that the \textsf{offline-Simon} algorithm of~\cite{DBLP:conf/asiacrypt/BonnetainHNSS19} can be extended to attack some symmetric constructions with a (provable) quantum time speedup 2.5. Our main example is the \emph{double-XOR Cascade construction} ({\sf 2XOR} in what follows) of~\autoref{fig:2xor}, introduced by Gazi and Tessaro~\cite{DBLP:conf/eurocrypt/GaziT12}.

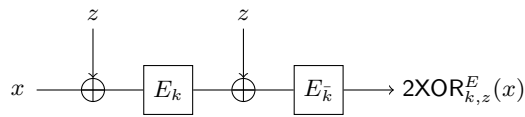
\begin{figure}[htbp]
\centering
\begin{tikzpicture}
  \draw
  node at (0,0)[l,name=m] {$x$}
  node [xor,circle, name=xor1, right of=m] {}
  node [block,name=e2, right of=xor1] {$E_k$}
  node [xor, circle, name=xor2, right of=e2] {}
  node [block,name=e3, right of=xor2] {$E_{\bar{k}}$}
  node [name=k1, above of=xor1] {$z$}
  node [name=k2, above of=xor2] {$z$}
  node [name=out, right of=e3, right] {$\mathsf{2XOR}_{k,z}^E(x)$};
  \draw[->] (m) -- (e2) -- (e3) -- (out);
  \draw[->] (k1) -- (xor1);
  \draw[->] (k2) -- (xor2);
 \end{tikzpicture}
\caption{The {\sf 2XOR} construction of~\cite{DBLP:conf/eurocrypt/GaziT12}. $E$ is an ideal $n$-bit block cipher, $z$ is an $n$-bit key, $k$ is a $\kappa$-bit key and $\bar{k}$ is $\pi(k)$ for some chosen permutation $\pi$ without fixpoints.}
\label{fig:2xor}
\end{figure}


From an $n$-bit block cipher with key length $\kappa$, the {\sf 2XOR} builds a block cipher with key length $n + \kappa$. It can be seen as a strengthening of the {\sf FX} construction (which would have a single block cipher call) that enhances the security when the adversary can make many queries. Indeed, in the ideal cipher model, any classical key-recovery of $\mathsf{2XOR}_{k,z}^E$ requires at least $\bigO{2^{\kappa + n/2}}$ evaluations of $E$, even in a regime where the adversary has access to the full codebook of $\mathsf{2XOR}_{k,z}^E$.

In the quantum setting, one can prove (see \appendixref{section:quantum-lb}) that a quantum adversary needs at least $\bigO{2^{\kappa/2}}$ quantum queries to either $E$, $\mathsf{2XOR}_{k,z}^E$ or their inverses. In~\autoref{section:quantum-attacks}, we show the following:
\begin{quotation}
Given $2^u$ classical chosen-plaintext queries to $\mathsf{2XOR}_{k,z}^E$, a quantum attacker can retrieve the key $k, z$ in quantum time $\bigO{n 2^u +  n^3 2^{(\kappa + n-u) / 2} }$.
\end{quotation}

In particular, when $\kappa = 2n$, a classical adversary knowing the full codebook needs a time $\bigO{}(2^{\frac{5n}{2}})$ to recover the key, whereas a quantum adversary requires only $\bigOt{2^{n}}$. In that case, $\mathsf{2XOR}_{k,z}^E$ offers actually \emph{no improvement} over the standalone cipher $E$, in the quantum setting.

\begin{figure}[htbp]
\centering
\begin{tikzpicture}
  \draw
  node at (0,0)[l,name=m] {$x$}
  node [block,name=e1, right of=m] {$E_k^1$}
  node [xor,circle, name=xor1, right of=e1] {}
  node [block,name=e2, right of=xor1] {$E_k^2$}
  node [xor, circle, name=xor2, right of=e2] {}
  node [block,name=e3, right of=xor2] {$E_k^3$}
  node [name=k1, above of=xor1] {$k_1$}
  node [name=k2, above of=xor2] {$k_2$}
  node [name=out, right of=e3, right] {$E_{k, k_1,k_2}(x)$};
  \draw[->] (m) -- (e1)-- (e2) -- (e3) -- (out);
  \draw[->] (k1) -- (xor1);
  \draw[->] (k2) -- (xor2);
 \end{tikzpicture}
\caption{\emph{Doubly-extended FX} ({\sf DEFX}) construction. $E^1, E^2, E^3$ are possibly independent block ciphers, but using the same $\kappa$-bit key $k$. $k_1$ and $k_2$ are independent $n$-bit whitening keys.}
\label{fig:constr}
\end{figure}
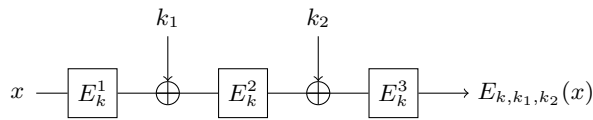

Beyond $\mathsf{2XOR}$, we use {\sf offline-Simon} to attack the extended construction of~\autoref{fig:constr} with the same complexity. 
We identify other settings where a quantum adversary can gain this 2.5 advantage, e.g., a key-recovery on {\sf ECBC-MAC} where part of the plaintext is unknown. We also extend our study to the case of \emph{known plaintext queries}, where all but a fraction of the codebook is known, and show that {\sf offline-Simon} still works in this setting.


This $2.5$ speedup was not observed before in~\cite{DBLP:conf/asiacrypt/BonnetainHNSS19} because the authors considered constructions such as {\sf FX}, which would omit the calls to $E_k^1$ and $E_k^3$. In that case, there exists improved classical time-data trade-offs that allow to reach precisely the square of the quantum time complexities, and \textsf{offline-Simon} only improves the memory consumption.

Whether this 2.5 speedup is the best achievable is an interesting question. We conjecture that variants of \textsf{offline-Simon} could reach a cubic speedup on appropriate problems, but we have not identified any corresponding cryptographic scenario.

\paragraph{Organization of the Paper.}
We start in \autoref{section:previous} by defining most of the block cipher constructions that will be considered in this paper, and their classical security results. We include results of quantum cryptanalysis for comparison. Details of the attacks are defferred to \autoref{section:prelim}, where we also cover some definitions and necessary background of quantum cryptanalysis, notably quantum search, Simon's algorithm and {\sf offline-Simon}.

We regroup our results and applications in \autoref{section:applis}. We introduce a construction similar to {\sf 2XOR} ({\sf EFX}) and propose self-contained proofs of classical and quantum security. Next, we detail our quantum attack in a chosen-plaintext setting. We also show that when almost all the codebook is known, known-plaintext queries can replace chosen-plaintext queries in {\sf offline-Simon}. This allows us to devise an attack against {\sf EFX} and a strengthened variant which we call {\sf DEFX}.

We discuss the limits of these results in \autoref{section:gap}. We conjecture that a variant of {\sf offline-Simon} could reach a cubic gap, though no corresponding cryptographic problem has been identified for now. We also discuss the apparent similarity of this 2.5 speedup with a 2.5 gap in query complexity~\cite{DBLP:conf/stoc/AmbainisBBLSS16}.

\paragraph{Notations.}
Throughout this paper, we will use the following notation:
\begin{itemize}
\item $E$ will be an $n$-bit state, $\kappa$-bit key block cipher: a family of $2^{\kappa}$ efficiently computable (and invertible) permutations of $\zo^n$. Security proofs consider the \emph{ideal model}, where $E$ is selected uniformly at random. Attacks (distinguishers, key-recoveries) are randomized algorithms whose success probability is studied on average over the random choice of $E$. We will also use $E^i$ to denote independent block ciphers.
\item $\Pi$ is a permutation of $\zo^n$, also selected uniformly at random.
\item $\omega$ is the matrix multiplication exponent. In practical examples, we can replace $\omega$ by 3 since the matrices considered are quite small (at most $256 \times 256$ for standard values of $n$).
\end{itemize}


\section{Classical Constructions and Previous Results}\label{section:previous}

In this section, we recapitulate the constructions considered in this paper. For each of them, we recall classical security bounds, quantum security bounds when they exist, and corresponding quantum attacks. These results are summarized in \autoref{table:attacks}. The quantum attacks will be detailed in \autoref{section:prelim}.

\subsection{Context}

We will use, for its simplicity, the Q1 / Q2 terminology of~\cite{DBLP:journals/tosc/KaplanLLN16,DBLP:conf/scn/HosoyamadaS18,DBLP:conf/asiacrypt/BonnetainHNSS19}, which is the most common in quantum cryptanalysis works. Alternative names exist, such as ``quantum chosen-plaintext attack'' (qCPA) instead of Q2, found in most provable security works (e.g.,~\cite{DBLP:conf/pqcrypto/AnandTTU16}) and~\cite{DBLP:conf/ctrsa/ItoHMSI19,DBLP:conf/indocrypt/CidHLS20}.

\begin{itemize}
\item A ``Q2'' attacker has access to a black-box quantum oracle holding some secret. We let $O_f$ denote a quantum oracle for $f$ (we will use the ``standard'' oracle representation, defined in \autoref{section:prelim}).
\item A ``Q1'' attacker can only query a black-box \emph{classically}. Naturally, Q2 attackers are stronger than Q1, since one can always emulate a classical oracle with a quantum one (it suffices to prepare the queries in computational basis states). The Q1 setting also encompasses any situation where there is no secret, for example preimage search in hash functions.
\end{itemize}

The constructions studied in this paper are block ciphers, studied in the \emph{ideal} (cipher or permutation) model. In particular, if $F = F_k[E]$ is the construction and $E$ is its internal component, we assume that $E$ is drawn uniformly at random, and let an attacker query $F$ and $E$ separately. The security proofs show lower bounds on the number of queries to $F$ and $E$ that an attacker must make to succeed. Such bounds can be proven for classical and quantum attackers alike. A Q2 attacker will have access to both $F$ and $E$ in superposition. Though a Q1 attacker will have only classical access to $F$, \emph{he still has quantum access to $E$}. Indeed, although supposedly chosen at random, $E$ remains a public component, with a public implementation. Thus, in the ideal model, Q1 attackers still make black-box quantum queries to $E$.

\paragraph{Attack Scenarios.}
Usually, an idealized cipher construction is proven to be a \emph{strong pseudorandom permutation} (\textsf{sPRP}, see \autoref{def:sprpadvantage} in \appendixref{sec:hcoefproof}). In this security notion, an adversary is asked to distinguish the construction $F_k[E]$ for a random $k$, from a random permutation, by making either forward or backward queries.

Obviously, a key-recovery attack is also a valid \textsf{sPRP} distinguisher. For all the constructions recalled in \autoref{table:attacks}, the security is proven with the \textsf{sPRP} game, and the attacks are key-recovery attacks.

\begin{table}[tb]
\caption{Summary of classical and quantum attacks considered in this paper. $D$ is the amount of classical queries to the construction. CPA = classical chosen-plaintext with classical computations. Q1 = classical chosen-plaintext with quantum computations (non adaptive). Q2 = quantum queries. KPA = classical known-plaintext. In quantum attacks, classical bits and qubits of memory are counted together for simplicity. We stress that all the quantum attacks considered here have only polynomial memory requirements. Complexities are displayed up to a constant. We do not consider attacks with preprocessing, or multi-user attacks. We assume $\kappa \geq n$.}\label{table:attacks}
\centering
\begin{tabular}{c|cccccc}
\toprule
Target  & Setting & Queries & Time & Mem. & Ref. \\
\midrule
{\sf EM}& Adaptive CPA& $2^{n/2}$ & $2^{n/2}$ & negl. &  \cite{DBLP:conf/eurocrypt/DunkelmanKS12}  \\
        & KPA & $D \leq 2^{n/2}$ & $2^n / D$ & $D$ &  \cite{DBLP:conf/eurocrypt/DunkelmanKS12} \\
        & Q2  & $n$ & $n^\omega$ & $n^2$ & \cite{DBLP:conf/isita/KuwakadoM12} \\
        & Q1  & $D \leq 2^{n/3}$ & $\sqrt{2^{n} / D}$ & $n^2$  & \cite{DBLP:conf/asiacrypt/BonnetainHNSS19} \\
\midrule
{\sf FX}& KPA & $D \leq 2^n$ & $2^{\kappa + n}/D$ & $D$ & \cite{DBLP:conf/eurocrypt/DunkelmanKS12} \\
& Adaptive CPA & $D \leq 2^{n/2}$ & $2^{\kappa + n}/D$ & negl. & \cite{DBLP:conf/eurocrypt/Dinur15} \\
& Adaptive CPA & $D \geq 2^{n/2}$ & $2^{\kappa + n}/D$ & $D^2 2^{-n}$ & \cite{DBLP:conf/eurocrypt/Dinur15} \\
& Q2 & $n$ & $n^\omega 2^{\kappa/2}$ & $n^2$ & \cite{DBLP:conf/asiacrypt/BonnetainHNSS19} \\
& Q1 & $D \leq 2^n$ & $\max(D, \sqrt{2^{\kappa + n}/D})$ & $n^2$ & \cite{DBLP:conf/asiacrypt/BonnetainHNSS19} \\
\midrule
{\sf 2XOR}& KPA & $D \leq 2^{n/2}$ & $2^{\kappa + n}/D$ & $D$ & \cite{DBLP:conf/eurocrypt/GaziT12} (adapted) \\
& Q2 & $n$ & $n^\omega 2^{\kappa/2}$ & $n^2$ & \autoref{section:applis} \\
& Q1 & $D \leq 2^n$ & $\max(D, \sqrt{2^{\kappa + n}/D})$ & $n^2$ & \autoref{section:applis} \\
\bottomrule
\end{tabular}
\end{table}

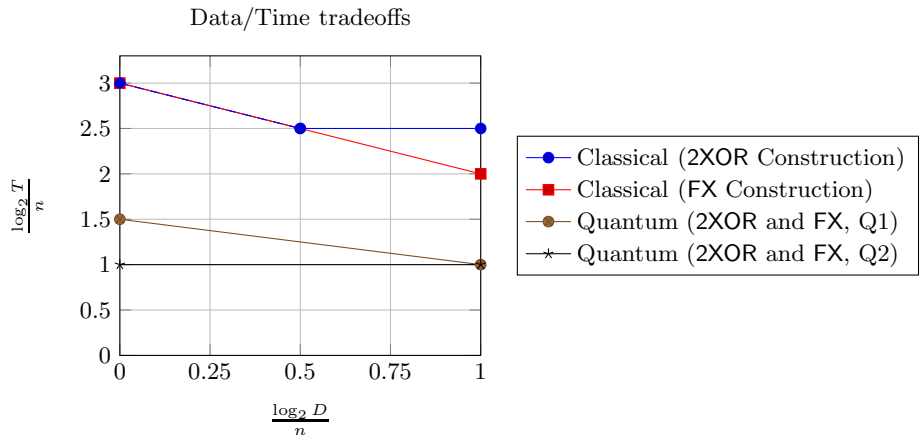
\begin{figure}[htbp]
\centering
\begin{tikzpicture}
\begin{axis}[
scale=0.7,
legend pos=outer north east,
xlabel={$\frac{\log_2 D}n$},
ylabel={$\frac{\log_2 T}n$},
xtick={0,0.25,...,1},
ytick={0,0.5,...,3},
ymin=0,xmin=0,xmax=1,
xmajorgrids,ymajorgrids,
title={Data/Time tradeoffs},
legend style={cells={anchor=west},name=legend,at={(1.1,0.5)},anchor=west}
]
\addplot coordinates {
(0, 3) (0.5, 2.5) (1, 2.5)
};
\addlegendentry{Classical ({\sf 2XOR} Construction)}
\addplot coordinates {
(0, 3) (1, 2)
};
\addlegendentry{Classical ({\sf FX} Construction)}
\pgfplotsset{cycle list shift=-2}
\addplot+[forget plot, dashed] coordinates {
(0, 3) (0.5, 2.5)
};
\pgfplotsset{cycle list shift=0}
\addplot coordinates {
(0, 1.5) (1, 1)
};
\addlegendentry{Quantum ({\sf 2XOR} and {\sf FX}, Q1)}
\addplot coordinates {
(0, 1) (1, 1)
};
\addlegendentry{Quantum ({\sf 2XOR} and {\sf FX}, Q2)}
\end{axis}
\end{tikzpicture}
\caption{Detail of \autoref{table:attacks}: comparison of the {\sf FX} and {\sf 2XOR} security in function of the number of queries for $\kappa = 2n$.}
\label{fig:}
\end{figure}

\subsection{The Even-Mansour Cipher}

The Even-Mansour cipher~\cite{DBLP:journals/joc/EvenM97} is a minimalistic construction which is ubiquitous in idealized designs. It starts from a public $n$-bit permutation $\Pi~: \zo^n \to \zo^n$ and two $n$-bit keys $k_1, k_2$ ($k_1 = k_2$ would be enough). The cipher is defined as: $\mathsf{EM}_{k_1, k_2}(x) = \Pi(x \oplus k_1) \oplus k_2$. If $\Pi$ is a random permutation, then an adversary making $T$ queries to $\Pi$ and $D$ queries to $\mathsf{EM}$ cannot recover the key with success probability more than $\bigO{ T D / 2^n}$. Matching attacks are known~\cite{DBLP:conf/asiacrypt/Daemen91,DBLP:conf/eurocrypt/DunkelmanKS12}. The quantum security was first studied by Kuwakado and Morii~\cite{DBLP:conf/isita/KuwakadoM12}, who gave a $\bigO{n^\omega}$ Q2 attack using $\bigO{n}$ queries (the attack will be presented later on). Several Q1 attacks were given in~\cite{DBLP:conf/isita/KuwakadoM12,DBLP:conf/ctrsa/HosoyamadaS18,DBLP:conf/asiacrypt/BonnetainHNSS19}. Only the latter (the most efficient) is displayed in~\autoref{table:attacks}.

\subsection{Key-length Extension Techniques}

Different ways of extending the key lengths of block ciphers have been proposed in the literature. Two well-known examples are the \emph{FX construction} and the \emph{Cascade} construction (or multiple-encryption).

\paragraph{FX-Construction.}
In~\cite{DBLP:conf/crypto/KilianR96}, Kilian and Rogaway proposed key whitenings as a solution to increase the effective key length of a block cipher $E$:
\[  \mathsf{FX}_{k_1, k_2, k}(x) =  E_k(x \oplus k_1) \oplus k_2 \enspace. \]
They showed that in the ideal model, an adversary making $D$ queries to $\mathsf{FX}$ needs to make $T = 2^{n + \kappa}/D$ to $E$ to recover the key. This is matched by the attacks of~\cite{DBLP:conf/eurocrypt/DunkelmanKS12,DBLP:conf/eurocrypt/Dinur15}.

The FX construction can also be seen as an Even-Mansour cipher where the public permutation $\Pi$ is replaced by an $n$-bit block cipher of unknown $\kappa$-bit key. This is why the attack strategies are similar. 

\paragraph{Quantum Security of FX.}
In~\cite{DBLP:journals/corr/abs-2105-01242}, it was shown that given $D$ \emph{non-adaptive} classical chosen-plaintext queries, a quantum adversary needs at least $\sqrt{2^{n + \kappa}/D}$ queries to $E$ to recover the key of $\mathsf{FX}$. This bound is matched by an attack of~\cite{DBLP:conf/asiacrypt/BonnetainHNSS19}, which is also non-adaptive. It seems likely that the same bound holds for adaptive queries, although this has not been formally proven.

\paragraph{Randomized Cascades.}
The \emph{double-XOR Cascade construction} (\textsf{2XOR}) was proposed in~\cite{DBLP:conf/eurocrypt/GaziT12}:
\[ \mathsf{2XOR}_{k,z}^E(m) = E_{\bar{k}} (E_k (m \oplus z) \oplus z)  \]
where $\bar{k}$ is $\pi(k)$ for some known fixpoint-free permutation $\pi$, $k$ is a $\kappa$-bit key and $z$ is an $n$-bit key. 

They show that if $E$ is an ideal cipher (drawn uniformly at random) and $k, z$ are chosen uniformly at random, then the \textsf{sPRP} advantage of an adversary making $q$ queries to $E$ is bounded by: $4 \left( \frac{q}{2^{\kappa + n/2}} \right)^{2/3}$ (Theorem~3 in~\cite{DBLP:conf/eurocrypt/GaziT12}). In particular, the adversary is free to query the whole codebook of $\mathsf{2XOR}_{k,z}^E$.

\paragraph{3XOR and 3XSK.}
Adding a third whitening key in the output of 2XOR yields the 3XOR construction of~\cite{DBLP:conf/fse/GaziLSST15}, which has an improved security. The authors also propose a construction without rekeying, where the two block ciphers are the same:
 \[  \mathsf{3XSK}_{k,z}[E](x) = E_k(E_k(x \oplus z) \oplus \pi(z) ) \oplus z \]
where $\pi$ is a permutation such that $z \mapsto z \oplus \pi(z)$ is also a permutation. As far as we know, the addition of the third whitening key actually renders the \textsf{offline-Simon} attack inoperable.


%



\section{Quantum Preliminaries}\label{section:prelim}

In this section, we recall some background of quantum cryptanalysis, going from Simon's algorithm to the \textsf{offline-Simon} algorithm from~\cite{DBLP:conf/asiacrypt/BonnetainHNSS19}. We assume that the reader is familiar with the basics of quantum computing~\cite{nielsen2002quantum} such as: the definitions of qubits, gates (Hadamard, Toffoli), quantum states and the ket notation $\ket{\psi}$. Note that we write quantum states without their global amplitude factors, e.g., $\frac{1}{\sqrt{2^n}} \sum_{x \in \zo^n} \ket{x}$ will be written $\sum_x \ket{x}$.

%
%
%

We will consider algorithms making oracle calls. A \emph{quantum} (or \emph{superposition}) oracle for a function $f$ will be represented as a black box unitary operator $O_f$: $O_f \ket{x} \ket{y} = \ket{x} \ket{y \oplus f(x)}$.

Any classical reversible algorithm $\mathcal{A}$ can be written as a circuit using only Toffoli gates. Then, there exists a quantum circuit $\mathcal{A}'$ that uses the same amount of gates, but instead of computing $\mathcal{A}(x)$ on an input $x$, it computes $\mathcal{A}$ in superposition: $\mathcal{A}' \ket{x} = \mathcal{A}(x)$. We call $\mathcal{A}'$ a \emph{quantum embedding} of $\mathcal{A}$. Classical algorithms are rarely written with reversibility in mind, but they can always be made reversible up to some trade-off between memory and time complexity overhead~\cite{DBLP:journals/siamcomp/Bennett89,DBLP:journals/siamcomp/LevinS90,DBLP:journals/corr/abs-math-9508218}.

\subsection{Quantum Search}

It is well known that Grover's algorithm~\cite{DBLP:conf/stoc/Grover96} provides a quadratic speedup on any classical algorithm that can be reframed as a black-box search problem. Amplitude Amplification~\cite{brassard2002quantum} further allows to speed up the search for a ``good'' output in any probabilistic algorithm, including another quantum algorithm.

Let $\mathcal{A}$ be a classical probabilistic algorithm with no input, and whose output has a probability $p$ to be ``good''; let $f$ a boolean function that effectively tests if the output is good. We are searching for a good output. 

Classical exhaustive search consists in running $\mathcal{A}$ until the output is good, and we will do that $\bigO{\frac{1}{p}}$ times. Quantum search is a \emph{stateful} procedure using $\bigO{\frac{1}{\sqrt{p}}}$ iterations of a quantum circuit that contains: a quantum implementation of $\mathcal{A}$, and a quantum implementation of $f$. In the case of Grover's algorithm, the search space is trivial, e.g., $\zo^n$. Here $\mathcal{A}$ has only to sample an $n$-bit string at random; the corresponding quantum algorithm is a Hadamard transform $H^{\otimes n} \ket{0} = \sum_{x \in \zo^n} \ket{x}$.

\begin{theorem}[From~\cite{brassard2002quantum}]\label{thm:aa}
Assume that there exists a quantum circuit for $\mathcal{A}$ using $T_A$ operations, and a quantum circuit for $f$ using $T_f$ operations. Then there exists a circuit ${\sf QSearch}(\mathcal{A}, f)$ that, with no input, produces a good output of $\mathcal{A}$. It runs in time: $\floor{ \frac{\pi}{4}  \frac{1}{\arcsin \sqrt{p}}} (2T_A + T_f)$ and succeeds with probability $\max\left(p, 1 - p\right)$.
\end{theorem}



\subsection{Simon's Algorithm}

In~\cite{DBLP:journals/siamcomp/Simon97a}, Simon gave the first example of an exponential quantum time speedup relative to an oracle.

\begin{problem}[Boolean period-finding]\label{pb:simon}
Given access to an oracle $f~: \zo^n \to \zo^m$ and the promise that:
\begin{itemize}
 \item ({Periodic case}) $\exists s \neq 0, \forall x, \forall y\neq x, [f(x) = f(y)\Leftrightarrow y = x \oplus s]$; or:
 \item ({Injective case}) $f$ is injective (i.e., $s = 0$).
\end{itemize}
Find $s$.
\end{problem}

Simon showed that when $f$ is a black-box classical oracle, this problem requires $\Omega(2^{n/2})$ queries, after which a classical adversary will find a collision of $f$, i.e., a pair $x,y$ such that $f(x) = f(y)$. He can then set $s = x \oplus y$ and verify his guess with a few more queries. 

However, given access to a quantum oracle $O_f$, a very simple algorithm solves this problem in $\bigO{n}$ quantum queries and $\bigO{n^\omega}$ classical postprocessing, where $\omega$ is the matrix multiplication exponent. This algorithm consists in repeating $\bigO{n}$ times a subroutine (\autoref{algo:simon}) which: $\bullet$~samples a random $n$-bit value $y$ in the injective case; $\bullet$~samples a random $n$-bit value $y$ such that $y \cdot s = 0$ in the periodic case. After $\bigO{n}$ samples, we can solve a linear system to find the case and recover $s$.

\begin{algorithm}[hbtp]
\begin{algorithmic}[1]
\State Start in the state \Comment{$\ket{0_n}\ket{0_m}$}
\State Apply a Hadamard transform \Comment{$\sum_x \ket{x} \ket{0_m}$}
\State Query $f$ \Comment{$\sum_x \ket{x} \ket{f(x)} $}
\State Measure the output register \label{step:output-measure}
\State Apply another Hadamard transform
\State Measure the input register, return the value $y$ obtained
\end{algorithmic}
\caption{Simon's subroutine.}\label{algo:simon}
\end{algorithm}

In the injective case, Step~\ref{step:output-measure} gives us a value $f(x_0)$ and makes the state collapse on $\ket{x_0}$ for some unknown $x_0$. The next Hadamard transform turns this into: $\sum_y (-1)^{x_0 \cdot y} \ket{y}$, and so, all $y$ are measured with the same probability

In the periodic case, the state collapses to a superposition of the two preimages $x_0$ and $x_0 \oplus s$: $\frac{1}{\sqrt{2}} \left( \ket{x_0} + \ket{x_0 \oplus s} \right)$. The next Hadamard transform turns this into:
\[ \sum_y \left( (-1)^{x_0 \cdot y} + (-1)^{(x_0 \oplus s) \cdot y} \right) \ket{y} \enspace, \]
and thus, the amplitudes of some of the $y$ turn to zero. These $y$ \emph{cannot} be measured. They are such that: $(-1)^{x_0 \cdot y} + (-1)^{(x_0 \oplus s) \cdot y} = 0 \implies s \cdot y = 1$, which means that we only measure random orthogonal vectors (besides, they all have the same amplitude).

\paragraph{Simon's Algorithm in Cryptanalysis.}
A typical example is the polynomial-time key-recovery on Even-Mansour of Kuwakado and Morii~\cite{DBLP:conf/isita/KuwakadoM12}. Given access to an Even-Mansour cipher $\mathsf{EM}_{k_1, k_2}$ of unknown key, define $f(x) = \mathsf{EM}_{k_1, k_2}(x) \oplus \Pi(x)$. It is periodic of period $k_1$. $\Pi$ is public, thus quantum-accessible. Given quantum oracle access to $\mathsf{EM}$, we can recover $k_1$.

Here, as most of the time in crypanalysis, the function $f$ cannot be promised to be \emph{exactly} injective or periodic, and additional collisions will occur.
Still, in our case, the output size of the periodic function is too large for these collisions to have any influence on the query cost~\cite{DBLP:journals/iacr/Bonnetain20}.

The same principle is used in most of the known quantum polynomial-time attacks in symmetric cryptography~\cite{DBLP:conf/isit/KuwakadoM10,DBLP:conf/isita/KuwakadoM12,DBLP:conf/crypto/KaplanLLN16,DBLP:conf/sacrypt/Bonnetain17,DBLP:conf/sacrypt/BonnetainNS19,DBLP:conf/asiacrypt/Leander017}. A cryptanalysis problem, such as the recovery of the key or of an internal value, is encoded as a period-recovery problem.

\subsection{Grover-meet-Simon}

In~\cite{DBLP:conf/asiacrypt/Leander017}, Leander and May proposed to combine Simon's algorithm with quantum search to attack the FX construction:
\[ {\sf FX}_{k, k_1, k_2}(x) = E_k(x \oplus k_1) \oplus k_2 \enspace.\]
Indeed, if we guess correctly the internal key $k$, then we can break the resulting Even-Mansour cipher. In fact, one can actually \emph{recognize the good $k$} by running an Even-Mansour attack: it will be successfull only with the correct $k$.

More generally, the \emph{Grover-meet-Simon} algorithm solves the following problem.

\begin{problem}\label{pb:gms}
Given access to a function $F(x,y)~: \zo^n \times \zo^\kappa \to \zo^n$ such that there exists a unique $y_0$ such that $F(\cdot, y_0)$ is periodic, find $y_0$ and the corresponding period.
\end{problem}

The algorithm is a quantum search over the value $y \in \zo^\kappa$. In order to guess a key $y$, it runs Simon's algorithm internally on the function $F(\cdot, y)$. It ends after $\bigO{n 2^{\kappa/2}}$ quantum queries to $F$ and $\bigO{n^\omega 2^{\kappa/2}}$ quantum time.

Having no interfering periods for \emph{all the functions} of the family $F(\cdot, y)$ allows to obtain an overwhelming probability of success for each test, and ensures the correctness of the algorithm. Again, this condition is satisfied for objects of cryptographic interest, and a tighter analysis if given in~\cite{DBLP:journals/iacr/Bonnetain20}. In the case of {\sf FX}, we define $F(x,y) = \mathsf{FX}_{k_1, k_2, k}(x) \oplus E_y(x)$.

\paragraph{Reversible Simon's Algorithm.}
Let us focus on the test used inside the {\sf FX} attack: it is a quantum circuit that, on input $\ket{y} \ket{0}$, returns $\ket{y} \ket{b}$ where $b = 1$ iff $x \mapsto F(y, x)= \mathsf{FX}_{k_1, k_2, k}(x) \oplus E_y(x)$ is periodic.

This quantum circuit first makes $c = \bigO{n}$ oracle queries to $F(y,x)$, building the state:
\begin{equation}\label{eq:gms-db}
\bigotimes_{1 \leq i \leq c} \sum_x \ket{x} \ket{ F(y, x)  } = \bigotimes_{1 \leq i \leq c} \sum_x \ket{x} \ket{ \mathsf{FX}_{k_1, k_2, k}(x) \oplus E_y(x)  } \enspace.
\end{equation}
These $c$ queries are all uniform superpositions over $x$, and require to query $\mathsf{FX}$. From this state, Simon's algorithm is run reversibly, without measurements. After a Hadamard transform, the input registers contain a family of $\bigO{n}$ vectors, whose dimension is computed. If the dimension is smaller than $n$, then the function is likely to be periodic.

We say ``likely'' because there is some probability to fail. These failures do not disrupt the algorithm, as shown in~\cite{DBLP:conf/asiacrypt/Leander017,DBLP:conf/asiacrypt/BonnetainHNSS19,DBLP:journals/iacr/Bonnetain20}.

These computations can be reverted and the state of~\autoref{eq:gms-db} is obtained again. It can now be reverted to $\ket{0}$ by doing the same oracle queries to $F(y,x)$.


\subsection{Offline-Simon}\label{section:offline-simon-q1}

The \textsf{offline-Simon} algorithm of~\cite{DBLP:conf/asiacrypt/BonnetainHNSS19} can be seen as an optimization of Grover-meet-Simon, where all queries to $ \mathsf{FX}_{k_1, k_2, k}$ are removed from the algorithm, except for the very first ones.

Crucially, the $\mathsf{FX}$ queries remain independent of the internal key guess $y$, and they are always made on the same uniform superposition $\sum_x \ket{x}$. Thus, we can consider that the following state:
\[ \ket{\psi} = \bigotimes_{1 \leq i \leq c} \sum_x \ket{x} \ket{  \mathsf{FX}_{k_1, k_2, k}(x) } \enspace,  \]
is given to the test circuit and returned afterwards. Intuitively, the state $\ket{\psi}$ stores all the data on $\mathsf{FX}$ that is required to run the attack, in a very compact way, since it fits in $\bigO{n^2}$ qubits.

With the queries done once beforehand and reused through the algorithm, the analysis is slightly different, but $\bigO{n}$ queries are still sufficient to succeed~\cite{DBLP:conf/asiacrypt/BonnetainHNSS19,DBLP:journals/iacr/Bonnetain20}.

\paragraph{Requirements.}
Not all Grover-meet-Simon instances can be made ``offline''. For this, we need the function $F(x,y)$ to have a special form, such as $F(x,y) = f(x) \oplus g(x,y)$ where $f$ ($\mathsf{FX}$ in our case) is be the \emph{offline} function, and $g$ ($E$ in our case) the \emph{online} one. In that case, to find the single $y_0$ for which $F(\cdot, y_0)$ is periodic, it suffices to make $\bigO{n}$ queries to $f$ at the beginning of the algorithm.


\paragraph{Offline-Simon and Q1 Attacks.}
As \textsf{Offline-Simon} uses only a polynomial number of queries, such queries can become very costly without significantly increasing the time cost of the algorithm. In particular, we can now replace the quantum queries by \emph{classical queries} and obtain interesting time-data trade-offs. We will keep the example of FX, taken from~\cite{DBLP:conf/asiacrypt/BonnetainHNSS19}, with a $\kappa$-bit internal key and a block size of $n$ bits. We assume that the adversary can make $D \leq 2^n$ chosen-plaintext queries to $\mathsf{FX}$.

With the \textsf{offline-Simon} algorithm, we proceed as follows. We let $D = 2^u$ for some $u$, and $k_1 = k_1^l \| k_1^r$, where $k_1^l$ is a subkey of $u$ bits. We define a function with a ``reduced codebook'':
\[ \begin{cases}
 G~: \zo^u \times \zo^{n-u} \times \zo^n \to \zo^n \\
 x,y_1,y_2 \mapsto \mathsf{FX}_{k_1, k_2, k}(x \| 0_{n-u} ) \oplus E_{y_2}(x \| y_1)
\end{cases}  \]
The key observation is that $G( \cdot, y_1, y_2)$ is periodic if and only if $y_1, y_2 = k_1^r, k$. In other words, part of the key will be handled by the quantum search, and part of it by the Simon subroutine.

We query $\mathsf{FX}_{k_1, k_2, k}(x \| 0_{n-u} )$ for all $x$. We use this data to produce ``manually'' the query states. This requires $\bigOt{2^{u}}$ operations, but \emph{in fine}, no Q2 queries at all. Next, the \textsf{offline-Simon} algorithm searches for the right value of $k_1^r, k$. This requires $\bigO{2^{(n + \kappa -u)/2}}$ iterations and $\bigO{n^\omega 2^{(n + \kappa -u)/2}}$ total time.

We end up with a time-data trade-off $D \cdot T^2 = \bigOt{2^{n + \kappa}}$, valid for $D \leq 2^n$. This means that for a given $D$, we get a time $T = \bigOt{\sqrt{ \frac{2^{n + \kappa}}{D} }}$, the square-root of the classical $T = \bigO{ 2^{n + \kappa} /D }$. However, while the classical attacks need $D$ memory, the quantum attack uses only $\bigO{n^2}$ qubits to store the database. This shows that Simon's algorithm is a crucial tool for this attack.



\section{New Result and Applications}\label{section:applis}


In this section, we show the 2.5 gap between a classical security proof (in the ideal model) and a quantum attack. Our target is a slightly more general construction than {\sf 2XOR}, that we denote by {\sf EFX}, for ``extended FX''.

\subsection{The EFX Construction and its Security}

Given two independent $n$-bit block ciphers $E^1, E^2$, of key size $\kappa$, and two $n$-bit whitening keys $k_1, k_2$, $\mathsf{EFX}_{k, k_1,k_2}[E^1, E^2]$ (or $\mathsf{EFX}_{k, k_1,k_2}$ for short) is an $n$-bit block cipher with $2n + \kappa$ bits of key (\autoref{fig:constr-simpl}):
\[ \mathsf{EFX}_{k, k_1,k_2}(x) = E_k^2 \left(k_2 \oplus E_k^1 (k_1 \oplus x) \right) \enspace. \]
It is a variant of 2XOR in which $E^1$ and $E^2$ are the same block cipher $E$, but under different keys $k, k' = \pi(k)$.


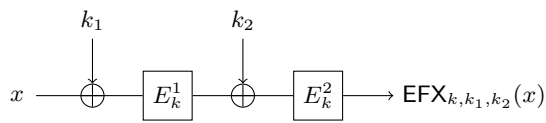
\begin{figure}[htbp]
\centering
\begin{tikzpicture}
  \draw
  node at (0,0)[l,name=m] {$x$}
  node [xor,circle, name=xor1, right of=m] {}
  node [block,name=e1, right of=xor1] {$E^1_k$}
  node [xor, circle, name=xor2, right of=e1] {}
  node [block,name=e2, right of=xor2] {$E^2_k$}
  node [name=k1, above of=xor1] {$k_1$}
  node [name=k2, above of=xor2] {$k_2$}
  node [name=out, right of=e2, right] {$\mathsf{EFX}_{k, k_1,k_2}(x)$};
  \draw[->] (m) -- (e1)-- (e2) -- (out);
  \draw[->] (k1) -- (xor1);
  \draw[->] (k2) -- (xor2);
 \end{tikzpicture}
\caption{The ``extended FX'' construction {\sf EFX}.}
\label{fig:constr-simpl}
\end{figure}

\paragraph{Classical Attack on EFX.}
The best attack on {\sf EFX} runs in time $\bigO{2^{\kappa + n/2}}$: one guesses the key $k$, then attacks the Even-Mansour cipher in time $2^{n/2}$.
In fact, this is the same classical attack as for the {\sf FX} construction with a slight change: after guessing the key, one has to perform reverse queries of the additional block cipher on the known ciphertext values.

Just like the attack on {\sf FX}, only $2^{n/2}$ known-plaintext queries are required for this (using the slidex attack on Even-Mansour~\cite{DBLP:conf/eurocrypt/DunkelmanKS12}).
However, having access to the whole codebook of {\sf EFX} does not seem to bring any improvement on the key-recovery since we'll still have to make matching queries to the additional block cipher.

More generally, let $D$ and $T$ be the number of online and offline queries respectively, the best attack runs in $DT = \bigO{2^{\kappa + n}}$  for $D \leq 2^{n/2}$ or else $T = \bigO{2^{\kappa + n/2}}$ for $D \geq 2^{n/2}$. 

\paragraph{Classical Proof of Security.}
The classical attack that we sketched above is essentially the best possible in the ideal cipher model. This can be deduced by the combination of the classical {\sf FX} security bound~\cite{JC:KilRog01} and the one derived by Ga{\v{z}}i and Tessaro~\cite{DBLP:conf/eurocrypt/GaziT12}. In~\autoref{sec:hcoefproof} we also give a new proof of \autoref{th:simplesec} that derives both of these bounds in a single go.

\begin{theorem}
  \label{th:simplesec}
  Consider the {\sf EFX} construction (\autoref{fig:constr-simpl}) and its \secnot{sPRP} game with $n$-bit state size and $\kappa$-bit ideal blockcipher key.
  An adversary $\mathcal{A}$ making $D$ online queries and $T$ offline queries has its advantage is bounded by both:
  \begin{align*} \adv_.^{\secnot{sprp}}(\mathcal{A}) \leq 
    &\frac{3}{2}\cdot\frac{T D}{2^{\kappa + n}}
    +\left(\frac{T^2 D}{2^{2\kappa + 2n}}\right)^{\frac{1}{3}}\\
  & + \left(\frac{2^{4n} T^{2}D}{2^{2(\kappa + 1)}(2^n-D+1)^3(2^n 
  -(T/D \cdot 2^{2n}/2^{\kappa})^{1/3} 
  -D+1)^3}\right)^{\frac{1}{3}}
  \end{align*}
  and:
  $$ \adv_.^{\secnot{sprp}}(\mathcal{A}) \leq  \frac{3}{2}\cdot\frac{T}{2^{\kappa + n/2}}$$
\end{theorem}

\begin{corollary}\label{corollary:classical-bound}
  Consider the {\sf EFX} construction (\autoref{fig:constr-simpl}) and its \secnot{sPRP} game.
  To obtain an $\Omega(1)$ advantage, it is required to have both $DT = \Omega(2^{\kappa + n})$ and $T = \Omega(2^{\kappa + n/2})$.
\end{corollary}

\paragraph{Quantum Proof of Security.}
In~\autoref{section:quantum-lb}, we study analogously the security in the \emph{quantum ideal cipher model}. We show that any quantum algorithm must make at least $\bigO{2^{\kappa/2}}$ queries to {\sf EFX} and its block ciphers to distinguish {\sf EFX} from a random permutation, with constant probability of success. Our attack matches the bound (up to a polynomial factor).

\subsection{Quantum Attacks}\label{section:quantum-attacks}

We can now explain how to attack {\sf EFX} in the quantum setting.


\begin{theorem}\label{thm:attack-efx}
There exists a quantum attack that, given $2^u$ classical chosen-plaintext queries to $\mathsf{EFX}$, finds the complete key $k, k_1, k_2$ of the cipher in quantum time $\bigO{n 2^u +  n^\omega 2^{(\kappa + n-u) / 2} }$. It succeeds with overwhelming probability when $E^1, E^2$ are chosen u.a.r.
\end{theorem}

\begin{proof}
The attack is very similar to the {\sf offline-Simon} attack on FX given in~\autoref{section:offline-simon-q1}. We write $k_1 = k_1^l \| k_1^r$ where $k_1^l$ is of $u$ bits and $k_1^r$ is of $n-u$ bits. We query the cipher on inputs of the form $x = * \| 0_{n-u}$, which take all $u$-bit prefixes, and are zero otherwise. We then use a quantum search over the complete key $k$ ($\kappa$ bits) and $k_1^r$.

The only difference with the {\sf FX} attack is in the way we test a guess $y_1, y_2$ of $k_1^r, k$. The database of queries now contains:
\[ \bigotimes_i \sum_{x \in \zo^u} \ket{x} \ket{\mathsf{EFX}(x \| 0_{n-u})} = \bigotimes_i \sum_{x \in \zo^u} \ket{x} \ket{ E_k^2( k_2 \oplus E_k^1( k_1^l \oplus x \| k_1^r)  } \enspace. \]
This means that given our guess $y_1, y_2$, we cannot just XOR the value of $E_{y_2}^1(x \| y_1)$ in place as we did before, because of the call to $E_k^2$.

Fortunately, since we have guessed $y_2$ (that is, the key $k$), we can map \emph{in place}:
\begin{multline*}
 \sum_{x \in \zo^u} \ket{x} \ket{ E_k^2( k_2 \oplus E_k^1( k_1^l \oplus x \| k_1^r)  } \\ \mapsto \sum_{x \in \zo^u} \ket{x} \ket{ (E^2_{y_2})^{-1} \left( E_k^2( k_2 \oplus E_k^1( k_1^l \oplus x \| k_1^r) \right) } \enspace, 
\end{multline*}
which, when $y_2 = k$, is exactly:
\[ \sum_{x \in \zo^u} \ket{x} \ket{ k_2 \oplus E_k^1( k_1^l \oplus x \| k_1^r) } \enspace.  \]
From there, we can XOR $E_{y_2}^1(x \| y_1)$ into the register and see if the function obtained is periodic. Both operations (the XOR and the permutation) are reversed afterwards, and we can move on to the next iteration.

While the periodic function can have additional collisions, its output size ($n$ bits) is actually larger than its input size ($u$ bits). Thus, with overwhelming probability, these collisions have no influence on the algorithm~\cite{DBLP:journals/iacr/Bonnetain20}. \qed
\end{proof}

In particular, when $\kappa = 2n$ and using $2^{n-1}$ classical queries, the attack would run in time $\bigO{n^\omega 2^n}$, compared to the classical $\bigO{2^{5n / 2}}$.

\begin{remark}
For a given $y_2$, $E_{y_2}$ is a permutation of known specification, of which we can compute the inverse. Thus the mapping $\ket{z} \mapsto \ket{E_{y_2}(z)}$ can be done in two steps using an ancillary register:
\[\ket{z} \ket{0} \mapsto \ket{z} \ket{E_{y_2}(z)} \mapsto \ket{z \oplus E_{y_2}^{-1}(E_{y_2}(z))} \ket{E_{y_2}(z)} = \ket{0}  \ket{E_{y_2}(z)}\enspace. \]
For more details on the implementation of such functions, see~\cite{DBLP:journals/iacr/Bonnetain20}.
\end{remark}

\begin{remark}
If the second block cipher call is done at the beginning, and not at the end, the same attack can be done with chosen-ciphertext queries.
\end{remark}

Let us note that within this attack, we are actually using {\sf offline-Simon} to solve the following problem.

\begin{problem}\label{pb:offline-simon-extended}
Given access to a function $f~: \zo^n \to \zo^n$ and a family of permutations $g_y~: \zo^n \to \zo^n$, indexed by $y \in \zo^\kappa$, such that there exists a single $y_0 \in \zo^\kappa$ such that $g_y(f)$ is periodic, find $y_0$.
\end{problem}

In the {\sf FX} attack, $g_y$ was the permutation: $x \mapsto g_y(x) = x \oplus E_y(x)$. Here we simply apply in place another block cipher call, before XORing.

\subsection{Attack with Known-Plaintext Queries}

The presentation of {\sf offline-Simon} in~\cite{DBLP:conf/asiacrypt/BonnetainHNSS19,DBLP:journals/iacr/Bonnetain20,DBLP:journals/iacr/BonnetainJ20}, which we followed in the previous section, constructs an \emph{exact} starting database, that is, a superposition of tuples $(x,f(x))$ with all $x$es forming an affine space. Note that to construct such a vector space, there are some constraints on the queries. There are three scenarios to efficiently achieve this:
\begin{itemize}
 \item The full codebook is queried,
 \item The queries are chosen,
 \item The queries are predictible and regular (for example, queries with a nonce incremented each time).
\end{itemize}

Hence, if we only have access to random known queries, we need to get the full codebook for our attack, which is a drastic limitation. In this section, we show that the algorithm still works if some values \emph{are missing}. That is, instead of:
\begin{multline*}
\ket{\psi}=  \bigotimes_{i = 0}^{c}  \sum_{x \in \zo^n} \ket{x} \ket{f(x) } , \\
 \text{ we start from } \ket{\psi'} = \bigotimes_{i = 0}^{c}  \sum_{x \in X} \ket{x} \ket{f(x)} + \sum_{x \notin X} \ket{x} \ket{0} \enspace, 
\end{multline*}
where $X \subsetneq \zo^n$ is the set of queries that we were allowed to make. In other words, we replace the missing output by the value 0.

Intuitively, if $X$ is close to $\zo^n$, the algorithm should not see that. It is actually easy to show by treating {\sf offline-Simon} as a black-box.

\begin{lemma}\label{lemma:approx-db}
Consider an instance of {\sf offline-Simon} with a starting database of $c = \bigO{n}$ states, that succeeds with probability $p$. Suppose that we now start from a database where a proportion $\alpha$ of queries is missing (that is, $|X| = (1-\alpha)2^n$). Then {\sf offline-Simon} still succeeds with probability at least $p (1 - \sqrt{2 c \alpha})^2$.
\end{lemma}

\begin{proof}
We can bound the distance between $\ket{\psi}$ and the $\ket{\psi'}$ defined above. Both are sums of $2^{n c}$ basis vectors with uniform amplitudes. There are less than $c \alpha 2^{n c}$ such vectors that appear in $\ket{\psi}$ and that do not appear in $\ket{\psi'}$, and vice-versa, as the value of $f(x)$ is incorrect in each $c$ states in $\ket{\psi'}$ for at most $\alpha2^n$ values. Thus:
\[\| \ket{\psi} - \ket{\psi'} \|^2  \leq  2 c \alpha \implies \| \ket{\psi} - \ket{\psi'} \| \leq \sqrt{2c \alpha}  \enspace. \]
Let $\ket{\phi}$ and $\ket{\phi'}$ be the states obtained after running \textsf{offline-Simon} with respectively $\ket{\psi}$ and $\ket{\psi'}$. We know that if we measure $\ket{\phi}$, we succeed with probability $p$. However, we are actually measuring $\ket{\phi'}$. We let $\ket{\phi_e} = \ket{\phi'} - \ket{\phi}$ the (non-normalized) error vector. We bound:
\[ \braket{ \phi | \phi_e }  \leq \| \ket{\phi} \| \| \ket{\phi'} - \ket{\phi} \| =  \| \ket{\psi'} - \ket{\psi} \|  \leq  \sqrt{2c \alpha}  \enspace,  \]
using the fact that a unitary operator (such as \textsf{offline-Simon}) preserves the euclidean distance. When measuring $\ket{\phi'}$, we project onto $\ket{\phi}$ with probability:
\[  (1 - \braket{ \phi | \phi_e } )^2 \geq (1 - \sqrt{2c \alpha})^2 \enspace, \]
and in that case we succeed with probability $p$.
\end{proof}

\begin{remark} If $\alpha = \bigO{1/n}$, then {\sf offline-Simon} succeeds with constant probability.
\end{remark}


Note that~\autoref{lemma:approx-db} only matters when we cannot choose the missing queries, i.e., in a known-plaintext setting. In a chosen-plaintext setting, it would always be more efficient to directly query an affine space.

\paragraph{Attack on EFX.}
Thanks to~\autoref{lemma:approx-db}, we can attack {\sf EFX} with known-plaintext queries provided that we have almost all the codebook, bypassing the need for a vector space in the inputs.

\begin{theorem}
There exists a quantum attack that, given $(1-\bigO{1/n})2^n$ classical \emph{known-plaintext} queries to $\mathsf{EFX}$, finds the complete key $k, k_1, k_2$ of the cipher in quantum time $\bigO{n 2^n + n^\omega 2^{\kappa / 2} }$.
\end{theorem}

In particular, we can also attack an even more generic version of {\sf EFX}, with three calls to independent block ciphers $E^1, E^2, E^3$. We call it {\sf DEFX}, for \emph{doubly-extended FX} (see~\autoref{fig:constr}):
\[ \mathsf{DEFX}(x) = E^3_k(k_2 \oplus E^2_k(k_1 \oplus E^1_k(x))) \enspace. \]
In this version, it suffices to remark that $\mathsf{DEFX}(x) = \mathsf{EFX}( E^1_k(x))$. We build states of the form $\sum_x \ket{x} \ket{\mathsf{DEFX}(x)}$ containing almost all the codebook. When we have guessed the right key $k$, we can map these states to:
\[ \sum_x \ket{ E^1_k(x) } \ket{ \mathsf{EFX}( E^1_k(x))} = \sum_{x'} \ket{x'} \ket{\mathsf{EFX}(x')} \enspace,  \]
by applying $E^1_k$ in place on the first register, and continue the attack as before.

\subsection{Applications}

The 2XOR-Cascade ({\sf 2XOR} for short) of~\cite{DBLP:conf/eurocrypt/GaziT12} is an instance of {\sf EFX}, and the results of~\autoref{section:quantum-attacks} immediately apply. This construction can also appear in other situations.

%

\paragraph{Encrypt-last-block-CBC-MAC with Unknown Plaintexts.}
{\sf ECBC-MAC} is an ISO standard~\cite[MAC algorithm 2]{isoecbcmac}, variant of {\sf CBC-MAC}, where the output of {\sf CBC-MAC} is reencrypted.

Let us consider a three-block {\sf ECBC-MAC} (\autoref{fig:ecbc}):
\[ m_0, m_1, m_2 \mapsto F(m_0, m_1, m_2) =  E_k'( E_k( m_2 \oplus E_k( m_1 \oplus E_k(m_0))) ) \enspace, \]
with a block cipher $E$ of $n$ bits, $2n$ bits of key $k$, and $k' = \phi(k)$ is derived from $k$. Assume that the adversary observes $F(m_0, m_1, m_2)$ for known values of $m_0$ (for example, a nonce) and \emph{fixed, but unknown} values of $m_1, m_2$.

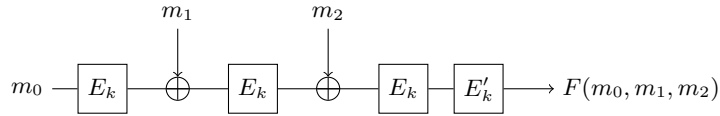
\begin{figure}[htbp]
\centering
\begin{tikzpicture}
  \draw
  node at (0,0)[l,name=m] {$m_0$}
  node [block, name=e0, right of=m] {$E_k$}
  node [xor,circle, name=xor1, right of=e0] {}
  node [block,name=e1, right of=xor1] {$E_k$}
  node [xor, circle, name=xor2, right of=e1] {}
  node [block,name=e2, right of=xor2] {$E_k$}
  node [block, name=e4, right of=e2] {$E_k'$}
  node [name=k1, above of=xor1] {$m_1$}
  node [name=k2, above of=xor2] {$m_2$}
  node [name=out, right of=e4, right] {$F(m_0, m_1, m_2)$};
  \draw[->] (m) -- (e0) -- (e1)-- (e2) -- (e4) -- (out);
  \draw[->] (k1) -- (xor1);
  \draw[->] (k2) -- (xor2);
 \end{tikzpicture}
\caption{Three-block {\sf ECBC-MAC}.}
\label{fig:ecbc}
\end{figure}

Then the problem of recovering $k, m_1, m_2$ altogether is equivalent to attacking a {\sf DEFX} construction where the cascade encryption with two different keys derived from $k$ is seen as another blockcipher with key $k$ : $E_k'(E_k(x)) = E^2_k(x)$.
More precisely, we assume that the adversary can query for $2^n (1- \alpha)$ values of $m_0$, where $\alpha = \bigO{1/n}$. In that case, \autoref{corollary:classical-bound} implies that any classical attack will require $\bigO{2^{5n/2}}$ computations.
Our quantum attack has a time complexity $\bigO{n^\omega 2^{n}}$.

This means that, up to a polynomial factor, it is no harder for the quantum adversary to recover the key of this {\sf ECBC-MAC} instance, although only the first block is known, than it would be in a chosen-plaintext scenario (where a direct quantum search of $k$ becomes possible).

This enhanced key-recovery attack applies as well if the first block is a nonce that the adversary does not choose (as soon as he is allowed $(1-\bigO{1/n})2^n$ queries).

\paragraph{Iterated Even-Mansour Ciphers.}
A natural setting where this construction will occur is with iterated Even-Mansour ciphers with $r$ rounds, such as the one represented in~\autoref{fig:iterated-em}. They have been considered in a variety of contexts. In particular, a classical cryptanalysis of all 4-round such ciphers with two keys $k_0, k_1$, for all sequences of $k_0$ and $k_1$, is given in~\cite{DBLP:conf/asiacrypt/DinurDKS14} (Table 2). For 4 rounds and two keys, \textsf{offline-Simon} does not seem to bring a more than quadratic improvement in any case. However, if the number of rounds increases, we can schedule the keys in order to reproduce an {\sf EFX} construction, for example with:
\[ k_0, k_0, k_1, k_0, k_1, k_0 \enspace. \]
Here the best classical attack seems to be guessing $k_0$, then breaking the Even-Mansour scheme, in time $2^{3n/2}$. By~\autoref{thm:attack-efx}, the quantum attack runs in time $\bigOt{2^{2n/3}}$ which represents a more-than-quadratic speedup.


\begin{figure}[htbp]
\centering
\begin{tikzpicture}
  \draw
  node at (0,0)[l,name=m] {$x$}
  node [xor,circle, name=xor1, right of=m] {}
  node [block,name=p1, right of=xor1] {$\Pi_1$}
  node [l,name=k1, above of=xor1] {$k_1$};
  \draw[->] (m) -- (xor1) -- (p1);
  \draw[->] (k1) -- (xor1);
  \draw node [xor,circle, name=xor2, right of=p1] {}
  node [block,name=p2, right of=xor2] {$\Pi_2$}
  node [l,name=k2, above of=xor2] {$k_2$};
  \draw[->] (p1) -- (xor2) -- (p2);
  \draw[->] (k2) -- (xor2);
  \draw node [xor,circle, name=xor3, right of=p2] {}
  node [block,name=p3, right of=xor3] {$\Pi_3$}
  node [l,name=k3, above of=xor3] {$k_3$};
  \draw[->] (p2) -- (xor3) -- (p3);
  \draw[->] (k3) -- (xor3);
  \draw node [xor, circle, name=xor4, right of=p3] {}
  node [name=k4, above of=xor4] {$k_4$}
  node [l,name=c, right of=xor4] {$y$};
  \draw[->] (p3) -- (c);
  \draw[->] (k4) -- (xor4);
 \end{tikzpicture}
\caption{An iterated Even-Mansour cipher with 4 keys. The $\Pi_i$ are independent $n$-bit permutations.}
\label{fig:iterated-em}
\end{figure}
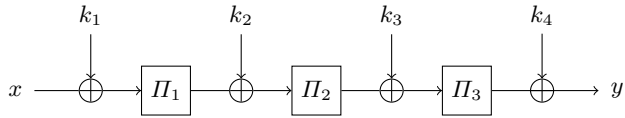

While such constructions have been proposed, they tend to avoid these unfavorable key schedules. The LED-128 block cipher~\cite{DBLP:conf/ches/GuoPPR11}, which can be analyzed as an iterated Even-Mansour scheme~\cite{DBLP:conf/asiacrypt/DinurDKS13}, alternates only between its two subkeys $k_0$ and $k_1$. Also, note that in these applications, the quantum attacks do not go below the classical query complexity lower bound ($\bigO{2^{nr/(r+1)}}$ for $r$-round Even-Mansour ciphers).

%


\section{On the Maximal Gap}\label{section:gap}

As we have recalled above, exponential speedups can be obtained when the quantum adversary can make superposition queries. For classical queries in symmetric cryptography, the best speedup remained quadratic for a long time. It is likely to remain polynomial, but as we manage to reach a 2.5 gap, it is natural to ask by how much we might extend it.

Note that if we formulate the question only as ``largest speedup when only classical queries are given'', it will not properly represent the class of symmetric cryptography attacks that we are interested in. Indeed, Shor's algorithm provides an exponential speedup on a problem with only classical queries.

However, there is still a major difference, in that we are interested in constructions \emph{with security proofs in the ideal model} (e.g., ideal ciphers, random oracles, random permutations). Here the definition of a \emph{largest gap} is more reasonable: all the quantum speedups known are polynomial at best. Besides, as we are interested in query complexities only, there could be some connection between classical and quantum query complexities. In this section, we discuss these possibilities.


%

%

\subsection{Limitations of Offline-Simon}

In this section, we will be interested in query complexity gaps only, and we do not make any considerations on the memory used or the time efficiency of the algorithms. 

In full generality, the \textsf{offline-Simon} algorithm can be seen as an algorithm that:
\begin{itemize}
\item queries a construction $F$ with unknown key, and populates a table with these queries
\item searches for some secret key $k$ using a quantum search where, in order to test a given $k$, queries to the table are made, and a \emph{superposition attack} on some construction is launched.
\end{itemize}

In particular, when attacking {\sf FX} with classical queries only, each iteration of the quantum search reproduces the attack on the Even-Mansour cipher -- and uses Simon's algorithm. But we could take this design more generally, and replace the Even-Mansour attack by any other attack using superposition queries. Thus, there is a link between the maximal gap achievable by the \textsf{offline} strategy and the maximal gap of superposition attacks.

\paragraph{Maximal Gap for Superposition Attacks.}
The gap in the Even-Mansour attack is $\mathsf{poly}(n)$ vs. $\bigO{2^{n/2}}$. We could try to increase it up to $\mathsf{poly}(n)$ vs. $\bigO{2^n}$. This is the best we can hope for, because we consider an n-bit construction: $\bigO{2^n}$ is its maximal query complexity.
All exponential speedups in quantum cryptanalysis that we know to date, including Q2 attacks on symmetric primitives, and attacks on asymmetric schemes, are based on variants of Simon's and Shor's algorithms. The classical counterpart of these algorithm is a collision search and, as such, they only reach a speedup $\mathsf{poly}(n)$ vs. $\bigO{2^{n/2}}$ at best. However, the quantum computing literature contains problems such as \emph{k-fold forrelation}~\cite{DBLP:journals/siamcomp/AaronsonA18}, with a gap $\bigO{2^{n - \epsilon}}$ against $\bigO{1}$ for every constant $\epsilon$, conjectured in~\cite{DBLP:journals/siamcomp/AaronsonA18} and proven in~\cite{DBLP:journals/eccc/BansalS20}. Using such a problem in place of Simon's algorithm could give us a cubic speedup (without a proof, but easy to conjecture). However, forrelation is not a problem that seems to naturally arise in cryptography. Finding a cryptographically relevant example of a gap between 2.5 and 3 is an interesting open question.

\paragraph{Composing the Speedups.}
We have seen that by composing quantum search, we cannot go above a global quadratic speedup. The same goes for any polynomial speedup. Though it would be possible to compose \textsf{offline-Simon} with itself (the superposition attack launched at each iteration contains then an instance of Grover-meet-Simon), doing so does not increase the speedup. 


%
%

\subsection{Relation with Query Complexity}

The question of finding the \emph{largest possible gap} in our context bears some similarities with the question of comparing randomized and quantum query complexities of total boolean functions. A gap 2.5 already exists in the related literature~\cite{DBLP:conf/stoc/AaronsonBK16}. We will now explain the reasons behind this coincidence.

\paragraph{Definitions.}
First of all, we need to recapitulate some essential definitions and results of query complexity. We will focus only on a very restricted subset of results. Let us consider a boolean function $f~: \zo^N \to \zo$. The definition of $f$ is known, and the only way to evaluate it is then to know some bits of its input string $x_0, \ldots, x_{N-1}$. Here, $N$ can be thought of as an exponential number.

When $f$ is defined over all its input, we call it a \emph{total} function, as opposed to a \emph{partial} function defined only over some domain $D \subseteq \zo^N$. For example, the $\mathsf{or}_N~: \zo^n \to \zo$ function computes the \textsf{OR} of all its bits.

For any $f$, we define:
\begin{itemize}
\item the deterministic query complexity $D(f)$: the minimum number of queries that have to be made by a deterministic algorithm computing $f(x)$ on every input $x$;
\item the bounded-error randomized query complexity $R(f)$: the minimum number of queries made by a randomized algorithm that outputs $f(x)$ with probability at least $2/3$ \emph{on every input $x$};
\item the quantum query complexity $Q(f)$: the minimum number of queries made by a quantum algorithm that outputs $f(x)$ with probability $2/3$.
\end{itemize}

For example, the classical query complexity of $\mathsf{or}_N$ is $N$, and its quantum query complexity is $\Theta(\sqrt{N})$ (thanks to Grover's algorithm and its matching lower bound).

Clearly, we have in general $Q(f) \leq R(f) \leq D(f)$. In classical cryptography we are usually interested in the measure $R(f)$, and in post-quantum cryptography in $Q(f)$. It has been known for a long time that for \emph{total} boolean functions, polynomial relations hold between these measures. In particular, Beals et al.~\cite{DBLP:journals/jacm/BealsBCMW01} showed that for any total function $f$, $D(f) = \bigO{Q(f)^6}$, and so $R(f) = \bigO{Q(f)^6}$. This was improved very recently in~\cite{DBLP:journals/eccc/AaronsonBKT20} to $D(f) = \bigO{Q(f)^4}$ (and so $R(f) = \bigO{Q(f)^4}$). The quartic relation with $D(f)$ is tight (by a separation given in~\cite{DBLP:conf/stoc/AmbainisBBLSS16}), but the relation with $R(f)$ is conjectured in~\cite{DBLP:conf/stoc/AmbainisBBLSS16} to be cubic.




\newcommand{\bigOmegat}[1]{\widetilde{\Omega}\left( #1 \right)}
\newcommand{\bigOmega}[1]{\widetilde{\Omega}\left( #1 \right)}

\paragraph{Promise Problems.}
These results underlie the idea that quantum speedups ``need structure'': indeed, an exponential quantum speedup can occur only if $f$ assumes some \emph{promise} on its input (for example for Simon's algorithm, that it encodes a periodic function).

Let us now take an example: the attack on {\sf EFX} of~\autoref{thm:attack-efx}. 

Recovering the key of an {\sf EFX} instance \emph{could} be seen as computing a boolean function $f$ with a promise. It would be done as follows: the input of the function encodes ${\sf EFX}_{k, k_1, k_2}[E](x)$ for all $x$ and $E_z(x)$ for all $(z,x)$; that is, the complete tables of the {\sf EFX} cipher and the ideal cipher upon which it is built. The function must compute the key $k, k_1, k_2$ used in ${\sf EFX}$. Although the second table ($E$) could be any value, since any block cipher can be selected at random, the function satisfies the promise that the first table actually encodes ${\sf EFX}[E]$.

There is, however, a significant difference between the query complexity of $f$ and the security of ${\sf EFX}$. The proof of security in the ideal cipher model reasons about adversaries as \emph{average-case} algorithms. Similarly, the classical and quantum attacks work \emph{on average} over all ciphers $E$. Typically, when running the Grover-meet-Simon attack, there are bad cases, corresponding to some rare choices of $E$, in which the algorithm will not be able to return the key. But the relations in query complexity concern only  worst-case complexities\footnote{Average-case complexities do not behave well, as shown in~\cite{DBLP:conf/stacs/AmbainisW00}.}. Our attack is not a worst-case algorithm, and so, we cannot say anything about $Q(f)$.

\paragraph{2.5 Separation Result.}
In~\cite{DBLP:conf/stoc/AaronsonBK16}, the authors proved the existence of a total function $f$ for which $R(f) = \bigOmegat{N^{2.5}}$ and $Q(f)= \bigOt{N}$. Since this 2.5 exponent is reminiscent of ours, we briefly review how it was obtained.

The authors start by defining a function with a promise, by composing \textsf{Forrelation} (a promise problem) and \textsf{And-Or} (a boolean function which has a provable quadratic quantum speedup). We do not need to define \textsf{Forrelation} here. Simon's problem could have been used instead, as a speedup $\mathsf{poly}(n)$ vs. $\bigO{2^{n/2}}$ is sufficient.

By combining \textsf{And-Or}s of size $N^2$ with a \textsf{Forrelation} of size $N$, one obtains a quantum algorithm running in time $Q(f) = \bigOt{N}$, because \textsf{Forrelation} requires $\bigOt{1}$ queries and \textsf{And-Or} requires $\bigO{N}$ queries using Grover's algorithm. The corresponding classical algorithm runs in time $\bigOt{N^{2.5}}$, due to the gap in both problems. Next, the authors introduce a generic \emph{cheat sheet} framework which allows to turn partial functions into total ones. The \emph{cheat sheet} variant of a function $f$, $f_{CS}$, is more costly. But this additional cost comes from a \emph{certificate} function, which checks if the input satisfies the promise. In the case studied in~\cite{DBLP:conf/stoc/AaronsonBK16}, the certificate simply consists in checking the outputs of the \textsf{And-Or}s, and checking that the \textsf{Forrelation} instance satisfies its promise: all of this can be done in quantum time $\bigOt{N}$. So the \emph{cheat sheet} variant of the above function provides the said query complexity gap.

The \textsf{offline-Simon} attack does actually the opposite of the function above. Instead of computing a \textsf{Simon} instance out of many individual \textsf{And-Or} results, it computes an \textsf{Or} of many independent \textsf{Simon} instances: we are looking for the single periodic function in a family of functions. This is why the 2.5 exponents coincide.

Besides, since we want to make only classical queries, we have to pay an additional cost $N$ corresponding to the classical queries to {\sf EFX}. This additional cost coincides with the cost of verifying the \textsf{Forrelation} instance. This is why, similarly to the cheat sheet technique, the {\sf offline-Simon} structure will allow a cubic gap at most. 
Yet, these are only similarities, as there is no connection between worst-case and average-case algorithms.



\section{Conclusion}

In this paper, we gave the first example of a more than quadratic speedup of a symmetric cryptanalytic attack in the classical query model. This 2.5 speedup is actually provable in the ideal cipher model. It is a direct counterexample to the folklore belief that \emph{doubling the key sizes} of symmetric constructions is sufficient to protect against quantum attackers. In particular, generic key-length extension techniques should be carefully analyzed: the {\sf 2XOR} Cascade proposed in~\cite{DBLP:conf/eurocrypt/GaziT12} offers practically no additional security in the quantum setting.

The most obvious open question is by how much this gap may be increased. The algorithm we used, \textsf{offline-Simon}, does not seem capable of reaching more than a 2.5 gap. Although a cubic separation seems achievable, we couldn't manage to obtain one with problems of cryptographic interest.
This is reminiscent of the cubic gap which is conjectured to be the best achievable between the randomized and quantum query complexities of total functions~\cite{DBLP:conf/stoc/AaronsonBK16}. However, there is a stark difference between the problems at stake, and in our case, it is not even known if a polynomial relation holds in general.

\ifeprint
\subsubsection*{Acknowledgements.}
The authors would like to thank Akinori Hosoyamada for helpful comments. A.S. has been supported by the ERC Advanced Grant N\textsuperscript{o} 740972 (ALGSTRONGCRYPTO).
\else
\fi

\bibliography{morethanq}
\bibliographystyle{splncs03}

\newpage
\appendix

\section*{\appendixname{}}


\section{Extended Quantum Search}\label{app:ext-qsearch}

We use the following recursive definition of algorithms that combine quantum searches.

\begin{definition}\label{def:quantum-search}
We define the class of ``extended quantum search algorithms'' (${\sf ExtQSearch}$), recursively, as follows:
\begin{itemize}
\item a Toffoli gate;
\item a quantum circuit formed from a sequence of ${\sf ExtQSearch}$ algorithms;
\item a quantum circuit ${\sf QSearch}(\mathcal{A})$ where $\mathcal{A}$ is in the class ${\sf ExtQSearch}$, and the test function of the quantum search is trivial (e.g., $\mathcal{A}$ outputs a boolean flag indicating if its result is correct).
\end{itemize}
\end{definition}

Similar definitions can be found in the literature, e.g.,~\cite{DBLP:journals/tosc/BonnetainNS19}. We use the fact that Toffoli gates are universal for reversible computing; any quantum circuit made only of Toffoli gates can be translated into a classical algorithm with the same time and memory complexities.

\begin{lemma}[Classical-quantum search correspondence]\label{lemma:quantum-search}
Let $\mathcal{A}$ be a quantum algorithm of the class ${\sf ExtQSearch}$. Assume that it uses a memory (counted in number of qubits) $\mathcal{M}$ and a time (in quantum gates) $\mathcal{T}$. Then there exists a probabilistic classical algorithm $\mathcal{A}'$ that emulates $\mathcal{A}$, i.e., returns the same results as if one measured the output of $\mathcal{A}$. This algorithm uses on average less than $\mathcal{T}^2$ logic gates and uses a memory of $\mathcal{M}$ bits.
\end{lemma}

\begin{proof}
In this proof, we will assume that \emph{quantum searches are exact}, meaning that after the number of iterations required by~\autoref{thm:aa}, the algorithm creates a uniform superposition of ``good'' outputs. This is not always the case, but it can only strengthen our result.

We prove the correspondence by writing down the algorithm $\mathcal{A}'$.

1. Assume that $\mathcal{A}$ is a Toffoli gate or a sequence of Toffoli gates, then this quantum circuit is actually a classical reversible circuit. The result follows trivially.

2. Assume that $\mathcal{A}$ is a sequence of calls to $m$ algorithms of the class ${\sf QSearch}$, denoted $\mathcal{A}_1, \ldots, \mathcal{A}_m$, of quantum time complexities $\mathcal{T}_1, \ldots, \mathcal{T}_m$ and maximal memory complexity $\mathcal{M}$. Assume that the correspondence holds for all the $\mathcal{A}_i$. Then there exist classical algorithms $\mathcal{A}_i'$ returning the same results in time $\mathcal{T}_1^2, \ldots, \mathcal{T}_m^2$ and with the same memory $\mathcal{M}$. The classical circuit for $\mathcal{A}'$ has the same layout as the quantum one, in which we replace all the $\mathcal{A}_i$ by the $\mathcal{A}_i'$. It can be easily seen that $\mathcal{A}'$ outputs the same distribution as $\mathcal{A}$. It has an average time complexity:
\[ \sum \mathcal{T}_i^2 \leq \left( \sum \mathcal{T}_i \right)^2 \enspace, \]
and uses the same memory as $\mathcal{A}$, since the $\mathcal{A}_i'$ have the same memories, and the circuit layout is the same.

3. Finally, assume that $\mathcal{A}$ is a quantum search: ${\sf QSearch}(\mathcal{B})$ where $\mathcal{B}$ is in the class ${\sf ExtQSearch}$ and allows a trivial test. Let $T_B$ be the quantum time complexity of $\mathcal{B}$, $p$ its success probability. By~\autoref{thm:aa}, the time complexity $T_A$ of $\mathcal{A}$ is greater than: $\floor{ \frac{\pi}{4} \frac{1}{\arcsin \sqrt{p}}} (2T_B )$.

For $x \in [0;1]$, we have $\arcsin x \geq x$, which implies $\frac{1}{\arcsin \sqrt{x}} \leq \frac{1}{\sqrt{x}}$. However $\frac{\pi}{2} \sqrt{x} \geq \arcsin \sqrt{x}$. Thus we can deduce that $T_A \geq \frac{1}{\sqrt{p}} T_B \implies T_A^2 \geq \frac{1}{p} T_B^2$. In the classical algorithm, we merely run $\mathcal{B}'$ until a good output is found. The average complexity is thus:
\[ T_{A'} = \frac{1}{p} T_{B'} \leq \frac{1}{p} T_B^2 \leq T_A^2 \enspace, \]
which proves the result.
\end{proof}

In particular for key-recovery attacks (the main focus of our paper), we can deduce the following corollary:
\begin{quotation}
If a cipher admits a quantum key-recovery attack (faster than Grover's algorithm) in the class ${\sf ExtQSearch}$, then it also admits a classical key-recovery attack.
\end{quotation}

Indeed, the classical procedure given by~\autoref{lemma:quantum-search} will have a time complexity below classical exhaustive search of the key. In other words, attacks based on quantum search fully comply with the paradigm of ``doubling the secret sizes''. It should be noted that this holds regardless of the query setting (Q1 or Q2).

This does not mean that these attacks are not interesting, since it is good to know if a quantum adversary can leverage the existing classical weaknesses of a cryptosystem. But it means that quantum search alone \emph{cannot introduce new breaks} regarding key-recovery.

\begin{remark}
We stress that this discussion concerns only key-recovery attacks. There are other quantum generic attacks that \emph{do not} offer a quadratic speedup, and in that case, procedures based on quantum search can effectively reduce the security margins classically established. This is the case in hash function cryptanalysis~\cite{DBLP:conf/eurocrypt/HosoyamadaS20,DBLP:conf/asiacrypt/DongSSGWH20,DBLP:journals/iacr/HosoyamadaS21,DBLP:journals/tosc/ChauhanKS21}.
\end{remark}

\newcommand{\efx}{\texttt{EFX} }

\section{Proving security}
\label{sec:hcoefproof}
\subsection{Security Game}

We want to prove the super Pseudorandom property of the \efx construction based on ideal ciphers.
That means we allow an hypothetic adversary to do forward and inverse queries to ideal cipher oracles as well as an encryption oracle.
In the real world, the three keys $k, k_1, k_2$ are first randomly drawn then the encryption oracle also makes use of the ideal cipher oracles to compute the output.
In the ideal world, a new permutation is randomly drawn and used to produce the output.
This is the \secnot{sPRP} security game as in Definition~\ref{def:sprpadvantage}.

\begin{definition}[sPRP Security] \label{def:sprpadvantage}
    Let $E_{1,k}(a)$ and $E_{2,k}(a)$ be two ideal ciphers with $\kappa$-bit key $k$ and $n$-bit input $a$, and $\mathcal{P}$ be the set of all $n$ to $n$ bit permutations.
    The sPRP security game advantage of an adversary for the \efx construction is defined as:
    \begin{align*}
        \adv_\efx^{\secnot{sprp}}(\mathcal{A}) & = \pr(\mathcal{A}^{E^{1/-1}_{\cdot, \cdot}(\cdot), \efx^{1/-1}_{k, k_1, k_2}(\cdot)} \rightarrow 1) - \pr(\mathcal{A}^{E^{1/-1}_{\cdot, \cdot}(\cdot),p^{1/-1}(\cdot)} \rightarrow 1)\;.
    \end{align*}
    with the randomness of ${k, k_1, k_2} \randfrom \{0,1\}^{\kappa + 2n}$, $p \randfrom \mathcal{P}$, the ideal ciphers $E_1$, $E_2$, and $\mathcal{A}$.

    Then, the sPRP security is the maximum advantage over all adversaries $\mathcal{A}$.
\end{definition}

\subsubsection{Transcript.}
As the adversary makes queries to the oracles we record the interactions in a transcript.
We denote $\mathcal{X}$ the set of all inputs of encryption queries and outputs of decryption queries with $D = |\mathcal{X}|$ the number of online queries.
Conversely, $\mathcal{Y}$ is the set of all outputs of encryption and input of decryption.
And $\mathcal{Q}^j_i$ is the set of all inputs of forward queries and output of backward queries to the ideal cipher $E_j$ parametrized with the key $i$ with $T^j_i = |\mathcal{Q}^j_i|$ and $T = \sum_{i \in \{0,1\}^\kappa; j \in \{1,2\}} T^j_i$ the total number of offline queries.

At the end of the interaction with the oracles, we help the adversary by providing additional information before the output decision.
Hence we define the final transcript $\tau$ as:
$$ \tau = \{k, k_1, k_2\} \cup \{(x, u, y), \forall x \in \mathcal{X}\} \cup \bigcup_{i \in \{0,1\}^\kappa; j \in \{1,2\}} \{(a, b), \forall a \in \mathcal{Q}^j_i\} $$
where $b = E_{j,i}(a)$ in both real and ideal worlds.
In the real world, 
\[y = \efx_{k, k_1, k_2}(x) = k_2 \oplus E_{2,k}(k_1 \oplus E_{1,k}(x))\] 
and, after interaction, we provide for the keys $k, k_1, k_2$ as well as the intermediary values $u = E_{1,k}(x)$ for all $x \in \mathcal{X}$.
In the ideal world, $y = p(x)$ that is the output of a randomly chosen permutation and we simulate the keys and intermediate values after interaction as in Algorithm~\ref{alg:idealtranscript}.

\begin{algorithm}
    \caption{Building Ideal Transcripts}
    \label{alg:idealtranscript}
	\begin{algorithmic}[1]
	\alginput $\{(x, p(x)), \forall x \in \mathcal{X}\} \cup \bigcup_{i \in \{0,1\}^\kappa} \{(a, E_{t,i}(a)), \forall a \in \mathcal{Q}^t_i, t \in \{1,2\}\}$\;.
	\algoutput $\{k, k_1, k_2\} \cup \{(x, u), \forall x \in \mathcal{X}\} $\;.
    \Procedure{IdealTranscript}{}
      \State $\{k, k_1, k_2\} \randfrom \{0,1\}^{n+2\kappa}$
      \State $\tau^\star \leftarrow \{k, k_1, k_2\}$
      \State $\mathcal{U} \leftarrow \emptyset$
      \ForAll{$a \in \mathcal{Q}_k^1$}
        \State $\mathcal{U} \leftarrow \mathcal{U} \cup \{E_{1,k}(a)\}$
      \EndFor
      \ForAll{$a \in \mathcal{Q}_k^2$}
        \If{$a \oplus k_1 \in \mathcal{U}$}
          \If{$E_{1,k}^{-1}(a\oplus k_1) \in \mathcal{X}$ \textbf{or} $\exists x \in \mathcal{X} : E_{2,k}(a) \oplus k_2 = p(x)$}
            \State \Return $\emptyset$ \Comment{Bad Event}
          \EndIf
        \Else
          \State $\mathcal{U} \leftarrow \mathcal{U} \cup \{a \oplus k_1\}$
        \EndIf
      \EndFor
      \ForAll{$x \in \mathcal{X}$}
      \If{$x \in \mathcal{Q}_k^1$ \textbf{and} $\exists a \in \mathcal{Q}_k^2 : E_{2,k}(a) = p(x) \oplus k_2$}
        \State \Return $\emptyset$ \Comment{Bad Event}

      \ElsIf{$x \in \mathcal{Q}_k^1$}
        \State $ \tau^\star \leftarrow \tau^\star \cup \{(x, E_{1,k}(x))\} $
      \ElsIf{$\exists a \in \mathcal{Q}_k^2 : E_k(a) = p(x) \oplus k_2$}
        \State $ \tau^\star \leftarrow \tau^\star \cup \{(x, a \oplus k_1)\} $
      \Else
        \State $u \randfrom \{0,1\}^n / \mathcal{U}$
        \State $\mathcal{U} \leftarrow \mathcal{U} \cup \{u\}$
        \State $ \tau^\star \leftarrow \tau^\star \cup \{(x, u)\} $

      \EndIf
      \EndFor
      \State \Return $\tau^\star$
    \EndProcedure
    \end{algorithmic}
\end{algorithm}

\subsection{H-coefficient Technique}

To prove \autoref{th:simplesec}, we will use the H-coefficient technique of \autoref{thm:hcoefficient}.

\begin{theorem}[H-coefficient technique] \label{thm:hcoefficient}
    Let $\mathcal{A}$ be a fixed computationally unbounded deterministic adversary that has access to either the real world oracle $\mathcal{O}_{\mathrm{re}}$ or the ideal world oracle $\mathcal{O}_{\mathrm{id}}$.
    Let $\Theta = \Theta_{\mathrm{g}} \sqcup \Theta_{\mathrm{b}}$ be some partition of the set of all attainable transcripts into \emph{good} and \emph{bad} transcripts.
    Suppose there exists $\epsilon_{\mathrm{ratio}} \geq 0$ such that for any $\tau \in \Theta_{\mathrm{g}}$,
    $$\frac{\pr(X_{\mathrm{re}} = \tau)}{\pr(X_{\mathrm{id}} = \tau)} \geq 1 - \epsilon_{\mathrm{ratio}}\,,$$
    and there exists $\epsilon_{\mathrm{bad}} \geq 0$ such that $\pr(X_{\mathrm{id}} \in \Theta_{\mathrm{b}}) \leq \epsilon_{\mathrm{bad}}$. Then,
    \begin{equation}\label{h-tech}
        \pr(\mathcal{A}^{\mathcal{O}_{\mathrm{re}}}\rightarrow1) - \pr(\mathcal{A}^{\mathcal{O}_{\mathrm{id}}}\rightarrow1) \leq \epsilon_{\mathrm{ratio}} + \epsilon_{\mathrm{bad}}\,.
    \end{equation}
\end{theorem}

\subsubsection{Bad Transcripts}

A transcript is said to be bad when Algorithm~\ref{alg:idealtranscript} return the empty set or when $T^1_k + T^2_k > \alpha T / 2^\kappa$ for some value $\alpha$ to be determined later.
Equivalently, a transcript is said to be bad when either $T^1_k + T^2_k > \alpha T / 2^\kappa$ or 
$\exists (a, b, x) \in  \mathcal{Q}^1_k \times \mathcal{Q}^2_k \times \mathcal{X} : a = b \oplus k_1, E_{1,k}(x) = a$ or 
$\exists (a, b, y) \in  \mathcal{Q}^1_k \times \mathcal{Q}^2_k \times \mathcal{Y} : a = b \oplus k_1, y = E_{2,k}(b)$ or
$\exists (x, a) \in (\mathcal{X} \cap \mathcal{Q}^1_k) \times \mathcal{Q}^2_k : E_{1,k}(a) = p(x) \oplus k_2$.

Firstly, we bound the probability of $T^1_k + T^2_k > \alpha T / 2^\kappa$ with the randomness of $k$ using the Markov inequality:
\begin{align}
  \pr( T^1_k + T^2_k > \alpha T / 2^\kappa ) \leq & 1 / \alpha
\end{align}

Then we bound the probability of $\exists (a, b, x) \in  \mathcal{Q}^1_k \times \mathcal{Q}^2_k \times \mathcal{X} : a = b \oplus k_1, E_{1,k}(x) = a$ with the randomness of $k$ and $k_1$:

\begin{align*}
  &\pr(\exists (a, b, x) \in  \mathcal{Q}^1_k \times \mathcal{Q}^2_k \times \mathcal{X} : a = b \oplus k_1, E_{1,k}(x) = a)\\
  = & \sum_{i \in \{0,1\}^\kappa} \pr(\exists (a, b, x) \in  \mathcal{Q}^1_i \times \mathcal{Q}^2_i \times \mathcal{X} : a = b \oplus k_1, E_{1,i}(x) = a) \pr( k = i ) \\
  \leq & 2^{-\kappa}\sum_{i \in \{0,1\}^\kappa} \min\left(\frac{\min(T^1_i, D) \cdot T^2_i}{2^n}, 1\right) \\
  \leq & 2^{-\kappa-n}\sum_{i \in \{0,1\}^\kappa} \min\left(T^1_i \cdot T^2_i, D \cdot T^2_i, 2^n\right)
\end{align*}
As we wish to get a born depending on $T$ but not on the repartition of the offline queries $T^j_i$, we assume the worst case that is the repartition giving the highest value.
Notice that $\sum_i (T^1_i \cdot T^2_i)$ can be optimized by maximizing a few terms.
In our case, we obtain an upper-bound by letting $T^1_i = T^2_i = \min(D, 2^{n/2})$ for $T/(2 \cdot \min(D, 2^{n/2}))$ different values of $i$ and $T^1_i = T^2_i = 0$ otherwise (the strategy of optimizing the values up to $\max(2^n/D, 2^{n/2})$ gives the same bound.):

\begin{align}
  \pr(\exists (a, b, x) \in  \mathcal{Q}^1_k \times \mathcal{Q}^2_k \times \mathcal{X} : a = b \oplus k_1, E_{1,k}(x) = a)
 \leq & \frac{T \cdot \min(D, 2^{n/2})}{2^{\kappa + n +1}}
\end{align}

We can derive the same bound the same way for the two remaining bad events $\exists (a, b, y) \in  \mathcal{Q}^1_k \times \mathcal{Q}^2_k \times \mathcal{Y} : a = b \oplus k_1, y = E_{2,k}(b)$ and
$\exists (x, a) \in (\mathcal{X} \cap \mathcal{Q}^1_k) \times \mathcal{Q}^2_k : E_{1,k}(a) = p(x) \oplus k_2$.
Putting it together:

\begin{align}
  \epsilon_{\mathrm{bad}} = \pr(\tau \text{ is bad}) & \leq \frac{1}{\alpha} + 3\frac{T \cdot \min(D, 2^{n/2})}{2^{\kappa + n +1}}
\end{align}

\subsubsection{Good Transcripts}

Assuming that $\tau$ is a good transcript, we want to upper-bound the ratio between the probabilities of $\tau$ happening in the real world and in the ideal world.
Let $\mathcal{A} = \{ E_{2,k}(a) \oplus k_2 : a \in \mathcal{Q}^2_k \} $ and $\mathcal{B} = \{ E_{1,k}(b) \oplus k_1 : b \in \mathcal{Q}^1_k \} $.

In the real world, the probability comes from the drawing of the keys and from every fresh queries to the ideal block cipher oracles:

\begin{align*}
&1/\pr(X_{\mathrm{re}} = \tau) = \\
&2^{\kappa + 2n}\left(\prod_{i \in \{0,1\}^\kappa / \{k\}; j \in \{1,2\}} (2^n)_{(T^j_i)}\right) \left( (2^n)_{(|\mathcal{Q}^1_k \cup \mathcal{X}|)} \cdot (2^n)_{(|\mathcal{A} \cup \mathcal{Y}|)} \right)
\end{align*}

In the ideal world, the probability comes from the ideal block cipher oracles, the encryption oracle and Algorithm~\ref{alg:idealtranscript}:
\begin{align*}
&1/\pr(X_{\mathrm{id}} = \tau) = \\
&2^{\kappa + 2n}\left(\prod_{i \in \{0,1\}^\kappa; j \in \{1,2\}} (2^n)_{(T^j_i)}\right)
\left((2^n)_{(D)} \cdot 
(2^n - |\mathcal{B} \cup \mathcal{Q}^2_k|)_{(D - |\mathcal{Q}^1_k \cap \mathcal{X}| - |\mathcal{A} \cap \mathcal{Y}|)}\right)
\end{align*}

Putting it together:
\begin{multline*}
  \frac{\pr(X_{\mathrm{re}} = \tau)}{\pr(X_{\mathrm{id}} = \tau)} \geq
  \frac{(2^n)_{(T^1_k)}(2^n)_{(T^2_k)}(2^n)_{(D)}(2^n - |\mathcal{B} \cup \mathcal{Q}^2_k|)_{(D - |\mathcal{Q}^1_k \cap \mathcal{X}| - |\mathcal{A} \cap \mathcal{Y}|)}}
  { (2^n)_{(|\mathcal{Q}^1_k \cup \mathcal{X}|)} (2^n)_{(|\mathcal{A} \cup \mathcal{Y}|)}}\\
  \geq
  \frac{(2^n)_{(T^1_k)}(2^n)_{(T^2_k)}(2^n)_{(D)}(2^n -T^1_k -T^2_k + |\mathcal{B} \cap \mathcal{Q}^2_k|)_{(D - |\mathcal{Q}^1_k \cap \mathcal{X}| - |\mathcal{A} \cap \mathcal{Y}|)}}
  { (2^n)_{(D + T^1_k - |\mathcal{Q}^1_k \cap \mathcal{X}|)} (2^n)_{(D + T^2_k - |\mathcal{A} \cap \mathcal{Y}|)}}\\
  \geq
  \frac{(2^n)_{(T^1_k)}(2^n)_{(T^2_k)}(2^n)_{(D)}(2^n -T^1_k -T^2_k)_{(D - |\mathcal{Q}^1_k \cap \mathcal{X}| - |\mathcal{A} \cap \mathcal{Y}|)}}
  { (2^n)_{(D + T^1_k - |\mathcal{Q}^1_k \cap \mathcal{X}|)} (2^n)_{(D + T^2_k - |\mathcal{A} \cap \mathcal{Y}|)}}
\end{multline*}

First notice that when $D = 2^n$ then $\mathcal{X} = \mathcal{Y} = \{0,1\}^n$ and we have $\frac{\pr(X_{\mathrm{re}} = \tau)}{\pr(X_{\mathrm{id}} = \tau)} \geq 1$.
Thus we can derive a first bound independent of $D$ ignoring the first bad event (or taking a very high $\alpha$):
$$ \adv_\efx^{\secnot{sprp}}(\mathcal{A}) \leq  3\frac{T}{2^{\kappa + n/2 +1}}$$

We have to work a bit more to get a bound for when $D \leq 2^{n/2}$:
\begin{align*}
  \frac{\pr(X_{\mathrm{re}} = \tau)}{\pr(X_{\mathrm{id}} = \tau)} 
  &\geq
  \frac{(2^n)_{(T^1_k)}(2^n)_{(T^2_k)}(2^n)_{(D)}(2^n -T^1_k -T^2_k)_{(D)}}
  { (2^n)_{(D + T^1_k)} (2^n)_{(D + T^2_k)}}\\
  &\geq
  \frac{(2^n)_{(D)}(2^n -T^1_k -T^2_k)_{(D)}}
  { (2^n -T^1_k)_{(D)} (2^n -T^2_k)_{(D)}}\\
  &\geq
  \left(\frac{(2^n -D +1)(2^n -T^1_k -T^2_k -D +1)}
  { (2^n -T^1_k -D +1) (2^n -T^2_k -D +1)}\right)^D\\
  &\geq
  \left(1 + \frac{T^1_k \cdot T^2_k}
  { (2^n -D +1)(2^n -T^1_k -T^2_k -D +1)}\right)^{-D}\\
  &\geq
  1 - \frac{D \cdot T^1_k \cdot T^2_k}
  { (2^n -D +1)(2^n -T^1_k -T^2_k -D +1)}\\
\end{align*}

Adding the fact that $T^1_k + T^2_k$ is upper-bounded by $\alpha T / 2^{\kappa}$ we get:

\begin{align}
  \epsilon_{\mathrm{ratio}} \leq \frac{\alpha^2 T^2 D}
  {2^{2\kappa + 2}(2^n-D+1)(2^n - \alpha T / 2^{\kappa} -D+1)}
\end{align}

\subsubsection{Conclusion}

Hence using the H-coefficient Technique of \autoref{thm:hcoefficient} we get two upper-bound for the advantage of a classical information theoretic adversary.
One mostly useful for $D \leq 2^{n/2}$:

$$ \adv_\efx^{\secnot{sprp}}(\mathcal{A}) \leq \frac{1}{\alpha} + 3\frac{T \cdot \min(D, 2^{n/2})}{2^{\kappa + n +1}} + \frac{\alpha^2 T^2 D}
{2^{2\kappa + 2}(2^n-D+1)(2^n - \alpha T / 2^{\kappa} -D+1)} $$

And one independent of $D$:
$$ \adv_\efx^{\secnot{sprp}}(\mathcal{A}) \leq  3\frac{T}{2^{\kappa + n/2 +1}}$$

Note that we are free to choose the value $\alpha$ to optimize the bound.
We decided to take $1/\alpha = \left(T^{2}D / 2^{2(\kappa + n)}\right)^{\frac{1}{3}}$ and that concludes the proof of \autoref{thm:attack-efx}.


\section{Quantum lower bounds}\label{section:quantum-lb}

In this section, we prove quantum lower bounds on the security of \efx{}, for completeness. The bounds we obtain are weak, in the sense that the security of the construct matches the security of its underlying primitive. However, they are also tight, as the \secnot{offline-Simon} algorithm makes for a matching upper bound.

We prove security in the \emph{quantum ideal cipher model}, introduced in~\cite{DBLP:conf/asiacrypt/HosoyamadaY18}. As for the ideal cipher model, we allow encryption and decryption queries to the block cipher $E^{\pm}$ and to the construction $C^{\pm}$. The only difference is that quantum queries are allowed, instead of classical queries. This means we prove security with quantum access to the construction $C^{\pm}$, which implies the bound when access to $C^{\pm}$ is only classical.

We note $\mathcal{C}(k,n)$ the distributions of $n$-bit block, $k$-bit key block ciphers, and $\mathcal{P}(n)$ the distribution of $n$-bit permutations. We will prove in this section the indistiguishability between $E^{\pm}, C^{\pm}$ with $E^{\pm}\in \mathcal{C}(k,n)$ and $E^{\pm}, P^{\pm}$ with $E^{\pm}\in \mathcal{C}(m,n)$ and $P^{\pm} \in \mathcal{P}(n)$ up to $2^{\kappa /2}$ queries.

We rely on the hardness of unstructured search:
\begin{lemma}[Optimality of Grover's algorithm{~\cite{Zalka99}}]\label{lemma:grover}
 Let $D_0$ be the degenerate distribution containing only the $\kappa$-bit input all-zero function,
 and $D_1$ be the distribution of $\kappa$-bit input boolean functions with only one output equal to 1. Then, for any quantum adversary $\mathcal{A}$ that does at most $q$ queries,
 \[
  \adv^{\secnot{dist}}_{D_0,D_1}(\mathcal{A}) \leq \frac{4q^2}{2^\kappa}\quad.
 \] 
\end{lemma}

We will now reduce the problem of distinguishing the construction from a random permutation to the unstructured search distinguisher.

\begin{lemma}[Distinguishing the \efx construction]
 Let $E_1, E_2 \randfrom \mathcal{C}(\kappa,n)^2$, $P\randfrom\mathcal{P}(n)$, $K_1,K_2\randfrom\zo^{\kappa+n}$, $\efx = E_2(K_1,E_1(K_1,x\oplus K_2) \oplus K_2)$. Then for any quantum adversary $\mathcal{A}$ that does at most $q$ queries,
 \[
  \adv^{\secnot{dist}}_{(E_1,E_2,\efx),(E_1,E_2,P}(\mathcal{A}) \leq \frac{4q^2}{2^\kappa}\quad.
 \]

\end{lemma}

\begin{proof}
 We want to reduce this distinguishing problem to the previous one. First, as in the classical proof, we can remark that the distributions of $E_1, E_2, \efx$ is equal to the distribution of $F, E_2, P$ with $E_2 \randfrom \mathcal{C}(k,n)$, $P\randfrom \mathcal{P}(n)$, $F(K_1,x) = P(E_2^{-1}(K_1,x\oplus K_2)\oplus K_2)$, and for other $K$, $F(K,x) \randfrom \mathcal{P}(n)$.
 
 Hence, we can consider the following construction: we take $E_1, E_2\in\mathcal{C}(\kappa,n)^2$, $P\in \mathcal{P}(n)$, $f\in \zo^\kappa\to \zo$, $K_2 \in \zo^n$. We can construct
 \[
  F(K,x) = \left\{\begin{tabular}{ll} $P(E_2^{-1}(K,x\oplus K_2)\oplus K_2)$ & {if $f(K) = 1$}\\
                                    $E_1(K,x)$ & {otherwise}
                  \end{tabular}
\right.\quad.
 \]
Now, we can leverage \autoref{lemma:grover} on the distribution of $(F, E_2, P)$: if $f$ is all-zero, we have the distribution of $(E_1, E_2, P)$. If $f$ has a unique 1, we have the distribution of $(E_1, E_2, \efx)$. Hence, any adversary that distinguishes $\efx$ can also distinguish unstructured search.
\end{proof}

\begin{remark}
 This proof can be directly adapted to the case where we only have one cipher, but two related keys are used.
\end{remark}

\begin{remark}[Tightness]
 This bound is tight when quantum query access is allowed. With only classical query access, the attack matches the bound only when $n \leq \kappa/2$. To prove security for smaller $n$ (or with a lower amount of classical data), one could adapt the quantum security proofs for the FX construction~\cite{DBLP:journals/corr/abs-2105-01242}, as the construction of interest is FX plus an additional encryption.
\end{remark}

\end{document}